\documentclass[letterpaper]{article} 
\usepackage{aaai25}  
\usepackage{times}  
\usepackage{helvet}  
\usepackage{courier}  
\usepackage[hyphens]{url}  
\usepackage{graphicx} 
\usepackage{natbib}  
\usepackage{caption} 
\usepackage{newfloat}
\usepackage{listings}
\DeclareCaptionStyle{ruled}{labelfont=normalfont,labelsep=colon,strut=off} 
\usepackage{amsmath}
\DeclareMathOperator*{\argmax}{arg\,max}
\usepackage{tikz}
\usetikzlibrary{shapes}
\usepackage{varwidth}
\usepackage{array,multirow,makecell}
\usepackage[switch]{lineno}

\usepackage[ruled,vlined,linesnumbered,noend]{algorithm2e}
\newcommand{\nosemic}{\renewcommand{\@endalgocfline}{\relax}}
\newcommand{\dosemic}{\renewcommand{\@endalgocfline}{\algocf@endline}}
\let\oldnl\nl
\newcommand{\nonl}{\renewcommand{\nl}{\let\nl\oldnl}}

\usepackage{todonotes}

\usepackage{pgfplots}
\usepackage{caption}
\usepackage{subcaption}
\usepackage{xcolor,colortbl}
\definecolor{LightCyan}{rgb}{0.88,1,1}
\definecolor{LightCyan2}{rgb}{0.48,0.69,0.99}
\definecolor{LightRed}{rgb}{0.99,0.44,0.43}
\definecolor{blueStrong}{rgb}{0.1,0.16,0.31}
\definecolor{linecolor}{rgb}{0.07,0.33,0.66}
\definecolor{bleu}{rgb}{0.07,0.33,0.76}
\definecolor{bleuvictoria}{rgb}{0.03,0.35,0.61}
\definecolor{bleudepths}{rgb}{0.15,0.39,0.32}
\definecolor{bleumalibu}{rgb}{0.12,0.25,0.96}
\definecolor{rouge}{rgb}{0.83,0.08,0}
\definecolor{racingred}{rgb}{0.74,0.09,0.17}
\definecolor{vert}{rgb}{0.12,0.74,0.13}
\definecolor{gold}{rgb}{1,0.84,0}
\definecolor{gold2}{rgb}{0.99,0.94,0.7}
\definecolor{aurometalsaurus}{rgb}{0.43, 0.5, 0.5}
\definecolor{cadetgrey}{rgb}{0.57, 0.64, 0.69}
\definecolor{lemon}{rgb}{0.94, 0.992, 0.37}
\definecolor{coffee}{rgb}{0.435, 0.305, 0.215}
\definecolor{liver}{rgb}{0.325, 0.294, 0.309}
\definecolor{rust}{rgb}{0.717, 0.155, 0.055}
\definecolor{darkkhaki}{rgb}{0.717, 0.155, 0.055}
\definecolor{tawny}{rgb}{0.74,0.71,0.42}
\definecolor{darkkhaki}{rgb}{0.817, 0.34, 0}
\definecolor{cocoabrown}{rgb}{0.208, 0.157, 0.118}
\definecolor{lilas}{rgb}{0.71, 0.4, 0.82}

\usepackage{amsthm}
\newtheorem{property}{Property}
\newtheorem{thm}{Theorem}

\newtheorem{theorem}{Theorem}

\newtheorem{observation}[theorem]{Observation}

\title{A Multiagent Path Search Algorithm for Large-Scale\\ Coalition Structure Generation}
\author {
    Redha Taguelmimt\textsuperscript{\rm 1},
    Samir Aknine\textsuperscript{\rm 2},
    Djamila Boukredera\textsuperscript{\rm 3},
    Narayan Changder\textsuperscript{\rm 4},\\
    Tuomas Sandholm\textsuperscript{\rm 5,6,7,8}
}
\affiliations {
    \textsuperscript{\rm 1}Clermont Auvergne University, Clermont Auvergne INP, CNRS, LIMOS, F-63000 Clermont-Ferrand, France.\\
    \textsuperscript{\rm 2}Univ Lyon, UCBL, CNRS, INSA Lyon, Centrale Lyon, Univ Lyon 2, LIRIS, UMR5205, Lyon, France\\
    \textsuperscript{\rm 3}Laboratory of Applied Mathematics, Faculty of Exact Sciences, University of Bejaia, Bejaia, Algeria\\
    \textsuperscript{\rm 4}TCG Centres for Research and Education in Science and Technology, Kolkata, India\\
    \textsuperscript{\rm 5}Carnegie Mellon University, Computer Science Department, Pittsburgh, USA\\
    \textsuperscript{\rm 6}Strategy Robot, Inc.\\
    \textsuperscript{\rm 7}Strategic Machine, Inc.\\
    \textsuperscript{\rm 8}Optimized Markets, Inc.\\
    redha.taguelmimt@gmail.com,
    samir.aknine@univ-lyon1.fr,
    djamila.boukredera@univ-bejaia.dz,
    narayan.changder@tcgcrest.org,
    sandholm@cs.cmu.edu
}

\usepackage{bibentry}

\begin{document}

\maketitle

\begin{abstract}
Coalition structure generation (CSG), i.e. the problem of optimally partitioning a set of agents into coalitions to maximize social welfare, is a fundamental computational problem in multiagent systems. This problem is important for many applications where small run times are necessary, including transportation and disaster response. 
In this paper, we develop SALDAE, a multiagent path finding algorithm for CSG that operates on a graph of coalition structures. Our algorithm utilizes a variety of heuristics and strategies to perform the search and guide it. It is an anytime algorithm that can handle large problems with hundreds and thousands of agents. We show empirically on nine standard value distributions, including disaster response and electric vehicle allocation benchmarks, that our algorithm enables a rapid finding of high-quality solutions and compares favorably with other state-of-the-art methods.
\end{abstract}

\section{Introduction}

\textit{Coalition formation} is a major problem in artificial intelligence that is central to many practical applications. Coalitions of delivery companies can, for instance, be formed to reduce transportation costs by sharing deliveries~\cite{sandholm1997coalitions}. In disaster response, hundreds of human responders can be quickly organized into teams to coordinate their evacuation and rescue actions~\cite{wu2020monte}. 
Central to coalition formation is the \textit{coalition structure generation} problem. It consists of identifying the optimal partitioning of a set of agents--- that is, an optimal set of coalitions among agents that maximizes the sum of the values of the coalitions (an optimal coalition structure).

Several algorithms have been proposed for this problem, including optimal and approximate solutions. Optimal solutions require a huge execution time and a large storage space due to the exponentiality of the input and the solution space. This limits their applicability to large-scale problems. Indeed, due to this, exact algorithms can only handle small numbers of agents (around 30)~\cite{wu2020monte}. Furthermore, we do not expect an algorithm that always finds an optimal solution to this problem for large-scale settings with hundreds of agents, as there is no proven guarantee that an algorithm can find an optimal solution without enumerating all $2^{n}$ coalition values. 
To meet the need of addressing this limitation, 
scalable solutions that scale to hundreds and thousands of agents have been suggested. Along this line, approaches such as CSG-UCT~\cite{wu2020monte}, 
PICS~\cite{10098066} and C-Link~\cite{farinelli2013c} have been proposed. However, to our knowledge, PICS, and CSG-UCT produce the best results among existing state-of-the-art algorithms for solving large-scale problem instances.

Given the complexity of this problem, we propose SALDAE (Scalable Algorithm with Large-Scale Distributed Agent Exploration)---a scalable and anytime algorithm for solving the \textit{coalition structure generation} problem. Specifically, we develop a variant of multiagent path finding for coalition structure generation by gradually building a search graph of coalition structures based on three expansion steps. To the best of our knowledge, this is the first algorithm for this problem using any variant of path finding. This algorithm is scalable and can be run with large problem instances with thousands of agents, which is hard to achieve with existing state-of-the-art exact algorithms. 
Moreover, it 
can return an anytime solution when time is limited. To summarize, our main contributions are: 

\begin{itemize}
    \item We present SALDAE, a new algorithm inspired by multiagent path finding concepts to search the coalition structure graph. It is anytime and scales to thousands of agents.
    \item To improve the search process, we propose various strategies inspired by MAPF (Multiagent path finding) techniques for connecting the best solutions found at a given step during the execution in an effort to find even better ones. We also explore different techniques for selecting child nodes and resolving conflicts between search agents, adapting ideas from MAPF to the CSG domain. Our results show that the use of these MAPF-inspired heuristics can significantly reduce the number of search steps required to find high-quality solutions.
    \item We empirically show that SALDAE outperforms the state-of-the-art algorithms for solving both small and large problems when generating anytime solutions.
\end{itemize}

\section{Problem Formulation}

In CSG, we are given a set of $n$ agents, represented by $\mathcal{A} = \{a_{1}, a_{2}, . . . , a_{n}\}$, and a characteristic function $v$ that assigns a real value to each coalition, indicating the efficiency of the coalition. A coalition $\mathcal{C}$ is any non-empty subset of $\mathcal{A}$, and the size of $\mathcal{C}$ is denoted by $|\mathcal{C}|$.
A coalition structure $\mathcal{CS}$ is a partition of the set of agents $\mathcal{A}$ into disjoint coalitions, formally defined as a collection of coalitions 
$\mathcal{CS} =\{\mathcal{C}_{1}, \mathcal{C}_{2},...,\mathcal{C}_{k}\}$, where $k= |\mathcal{CS}|$, and the following constraints are satisfied:  $\bigcup_{j=1}^{k} \mathcal{C}_{i} = \mathcal{A}$ (all agents are included in the coalition structure) and for all $i,j \in \{1,2,...,k\}$ where $i \neq j$, $\mathcal{C}_{i} \cap \mathcal{C}_{j} = \emptyset$ (each agent is included in exactly one coalition).
\par
Let $\Pi(\mathcal{A})$ denote the set of all coalition structures. The value of a coalition structure $\mathcal{CS}$ is assessed as the sum of the values of the disjoint coalitions that comprise it: $v (\mathcal{CS}) = \sum_{\mathcal{C} \in \mathcal{CS}} v(\mathcal{C})$. 
The CSG problem aims at finding the optimal solution, which is the most valuable coalition structure $\mathcal{CS}^{*} \in \Pi(\mathcal{A})$, i.e. $\mathcal{CS}^{*}=  $argmax$_{\mathcal{CS} \in \Pi(\mathcal{A})} v(\mathcal{CS})$, for a given set of agents $\mathcal{A}$. However, when time is limited having a good-enough quality solution within a reasonable time is more valuable.

The coalition structure graph, first introduced by ~\cite{sandholm1999coalition}, is a way to represent the search space as a graph composed of nodes representing the coalition structures. For a given set of $n$ agents, these nodes are organized into $n$ levels, where each level consists of nodes representing coalition structures that contain exactly $i$ coalitions ($i \in \{1,..,n\}$). Each edge of this graph connects two nodes belonging to two consecutive levels, such that each coalition structure at level $i$ can be obtained by dividing a coalition from a coalition structure at level $i - 1$ into two coalitions.

\section{Related Work}

Many approaches have been proposed to solve the CSG problem either optimally or approximately, including dynamic programming algorithms, anytime algorithms, heuristic algorithms, and scalable solutions. 
Dynamic programming approaches, such as those proposed in~\cite{yeh1986dynamic,rahwan2008improved,michalak2016hybrid}, guarantee to find the optimal coalition structure but must be run to completion to do so. 
Anytime algorithms~\cite{sandholm1999coalition,dang2004generating,rahwan2009anytime}, on the other hand, allow for premature termination while providing intermediate solutions during execution. 
The hybrid algorithms that combine dynamic programming algorithms with anytime algorithms~\cite{michalak2016hybrid,changder2020odss,Changder_Aknine_Ramchurn_Dutta_2021,ijcai2024p27,ijcai2023p35}  are the fastest exact algorithms. 

Heuristic algorithms, such as those proposed in~\cite{sen2000searching,keinanen2009simulated,di2010coalition}, prioritize speed and do not guarantee to find an optimal solution. These algorithms are useful when the number of agents increases and the problem becomes too hard to solve optimally. For instance, the simulated annealing method~\cite{keinanen2009simulated}, a stochastic local search algorithm, explores different neighborhoods of coalition structures by splitting, merging, or shifting agents. It starts with a random coalition structure and moves to a new one in its neighborhood at each iteration, with the movement accepted with a probability that depends on the difference in utility and a decreasing temperature parameter. 

Very few scalable solutions to CSG exist. For example, the Monte Carlo tree search method proposed in~\cite{wu2020monte} finds solutions by sampling the coalition structure graph and partially expanding a search tree that corresponds to a partial search space that has been explored. The hierarchical clustering approach proposed in~\cite{farinelli2013c} builds a coalition structure by merging coalitions using a similarity criterion based on the gain that the system achieves if two coalitions merge. The search algorithms FACS and PICS proposed in~\cite{9643288,10098066} 
generate coalition structures based on code permutations applied to selected initial vectors of a different search space representation. To the best of our knowledge, PICS~\cite{10098066} and CSG-UCT~\cite{wu2020monte} are the best performing of the prior algorithms. 
Our proposed algorithm evaluates possible splits of coalitions and possible merges of coalition pairs. Unlike some other approaches, such as the C-Link method~\cite{farinelli2013c}, which cannot split coalitions once they are merged, our algorithm allows for coalitions to be split and merged multiple times. This makes our algorithm less likely to get trapped in local maxima. 


Multiagent path finding (MAPF)~\cite{stern2019multi2} is another important problem in multiagent systems, where agents must navigate through a given environment to reach their goals while avoiding collisions with each other. This problem of planning paths for multiple agents is also known to be NP-hard~\cite{Yu_LaValle_2013} and has been studied extensively in the literature. It is inspired by real-world applications such as warehouse
logistics~\cite{Ma2017LifelongMP}, autonomous aircraft-towing vehicles~\cite{Morris2016PlanningSA}, airport operations~\cite{Li2019SchedulingAA}, and video games~\cite{li2020moving}. 
One of the most popular algorithms for MAPF is the Conflict-Based Search (CBS) algorithm~\cite{sharon2015conflict}. CBS is a complete algorithm that guarantees to find the optimal solution if one exists. Other algorithms~\cite{Gange_Harabor_Stuckey_2021,Barer2014SuboptimalVO,Li_Tinka_Kiesel_Durham_Kumar_Koenig_2021,li2021eecbs,Li_Gange_Harabor_Stuckey_Ma_Koenig_2020,Li2020EECBSAB} based on CBS and other methods have been developed to solve this problem in optimal, suboptimal, or bounded suboptimal ways. 
In the next section, we explain how we draw inspiration from MAPF concepts to enhance the search for solutions to CSG and propose a new algorithm that adapts ideas from MAPF for solving the CSG problem.

In this paper, we propose a path search algorithm for finding optimal coalition structures in coalition structure generation. The algorithm operates on a graph where each node represents a coalition structure, and a search agent explores this graph aiming to reach the highest-valued solution. The algorithm starts from a designated start node and iteratively explores neighboring nodes, guided by the values of coalition structures.

The key components of this approach that will be detailed in the following sections are:

\begin{itemize}
    \item The search space is represented as a graph, with each node corresponding to a coalition structure. Transitions between nodes occur through coalition splitting or merging.
    \item A search agent explores the graph by moving from one node to another, selecting nodes based on their coalition structure values.
    \item Upon finding better solutions, the algorithm creates paths between previous and new solutions to explore potential improvements along the path.
    \item To optimize search efficiency, the algorithm employs memory management techniques, maintaining lists of nodes in memory and dynamically adjusting them based on solution quality.
    \item Multiple search agents are employed to accelerate finding high-quality solutions, with conflict resolution mechanisms to prevent redundant exploration.
\end{itemize}

\section{Path Search Algorithm for CSG}

In this paper, we consider a path finding variant where each node is a solution to the CSG problem, i.e. a coalition structure. The goal is thus to find the optimal one or at least approach its value. This variant is defined by a graph and a search agent that begins at a start node and can move to an adjacent node at each step. The search agent maintains a list of nodes, sorted according to the values of the corresponding coalition structures. 
A path in this context is a sequence of nodes that are adjacent to each other, starting at the start node. 
The decision to pursue one path over another is based on the coalition structure value of the reached node. Hence, the cost of a path is the value of the coalition structure of its last expanded node. 

Multiagent Path Finding~\cite{stern2019multi,stern2019multi2} also has many variants. Our goal is to find the optimal solution using multiple search agents. Inspired by concepts from MAPF, our variant is defined by a graph and a set of $m$ search agents $\{s_{1},..., s_{m}\}$. Each search agent $s_{i}$ has a designated start node and can move to an adjacent node at each step. Conflicts can occur when two search agents consider the same coalition structure for evaluation. 
While the paths taken by the search agents are important for finding high-quality solutions quickly, they are not the solution themselves. 
The optimal solution is the highest-valued coalition structure found in the nodes.

Our path search algorithm for the optimal coalition structure 
uses a search graph (see Figure \ref{CSGSALDAE}). 
This graph is built gradually, starting from a designated start node that represents a starting coalition structure, which can be the top node (that represents the coalition structure containing the singleton coalitions), the bottom node (i.e. the coalition structure composed of the grand coalition that contains all the agents) or any other node. 
Each parent node in the search graph is connected to a number of child nodes that can be generated by either splitting a coalition or joining two coalitions in the parent coalition structure. 
Given a coalition structure $\mathcal{CS} =\{\mathcal{C}_{1}, \mathcal{C}_{2},...,\mathcal{C}_{k}\}$: 
\begin{itemize}
    \item splitting a coalition $\mathcal{C}_{i}$ into $\mathcal{C}_{j}$ and $\mathcal{C}_{k}$ (i.e. $\mathcal{C}_{i} = \mathcal{C}_{j} \bigcup \mathcal{C}_{k}$ with $\mathcal{C}_{j},\mathcal{C}_{k} \neq \emptyset$ and $\mathcal{C}_{j} \cap \mathcal{C}_{k} = \emptyset$) in $\mathcal{CS}$ results in a new coalition structure $(\mathcal{CS} \setminus \{\mathcal{C}_{i}\}) \bigcup \{\mathcal{C}_{j},\mathcal{C}_{k}\}$. 
    \item merging two coalitions $\mathcal{C}_{i}$ and $\mathcal{C}_{j}$ in $\mathcal{CS}$, with $i \neq j$, results in a new coalition structure $(\mathcal{CS} \setminus \{\mathcal{C}_{i},\mathcal{C}_{j}\}) \bigcup \{\mathcal{C}_{i} \bigcup \mathcal{C}_{j}\}$.
\end{itemize}

Each node in the search graph represents a potential solution. For the sake of clarity, we will consider the bottom node as the start node in the remainder of this section. The search graph is constructed iteratively by adding new nodes generated according to the following steps, which constitute a type of greedy algorithm. 
Throughout this process, the terms "node", "coalition structure", and "solution" may be used interchangeably.

\begin{enumerate}
\item \textit{Step 1: Generation}: 
In this step, child coalition structures are generated from the current start node by either splitting a coalition or joining two coalitions. For the first iteration, if the bottom node is the start node, child nodes can only be generated by splitting the grand coalition. Each newly generated child node is added to the search graph. 
Figure \ref{CSGSALDAE} shows an example with 4 agents. The numbers in Figure \ref{CSGSALDAE} represent the agents. For example, $\{1,2,3,4\}$ represents the coalition structure $\{a_{1},a_{2},a_{3},a_{4}\}$.

\item \textit{Step 2: Selection}: 
This step defines the start node selection procedure. When a set of child nodes is generated, they all form a set of candidate coalition structures to be the start node for the next iteration. Given this, the algorithm selects the most promising node to consider as the start node $\mathcal{SN}$ for the next iteration. The start node is chosen as the highest-valued child coalition structure, 
i.e. 
$\mathcal{SN} = \argmax_{\mathcal{CS} \in Child} V(\mathcal{CS})$, where $Child$ is the current child nodes generated in the graph, not just the child nodes generated for one node. Figure \ref{CSGSALDAE} illustrates this procedure. For this example, we consider the node $\{\{a_{1},a_{4}\},\{a_{2},a_{3}\}\}$ as the highest valued child coalition structure. All the nodes in level 2 were candidates for this selection (see Figure \ref{CSGSALDAE}).

\item \textit{Step 3: Comparison to incumbent}: 
This step involves evaluating the selected coalition structure to determine if it is better than the current best solution found in previous iterations. If it is, the best solution is updated with the new one. The value of a coalition structure is determined by summing the values of the coalitions that comprise it. After evaluating the coalition structure, another iteration begins with the new start node, which is the tested coalition structure. 
Figure \ref{CSGSALDAE} shows a 4-agent illustration example.

\end{enumerate}

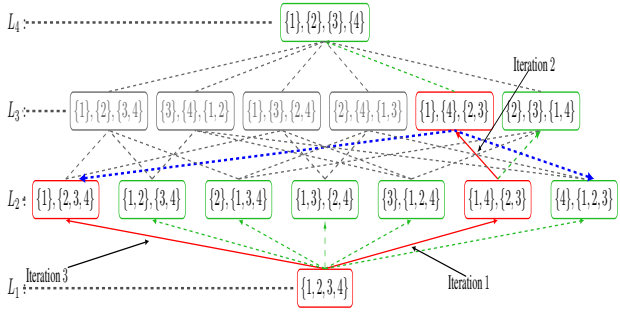
\begin{figure}[t]
\centering
\normalsize
\resizebox{0.97\columnwidth}{4.05cm}{%
\begin{tikzpicture}
\large

\node[] (a)at(-3.2,-4.5){$L_{1}: $};
\draw[>=latex,black!70,dashed,very thick] (6.4,-4.5) -- (-3.0,-4.5);
\node[] (a)at(-3.2,-3){$L_{2}: $};
\draw[>=latex,black!70,dashed,very thick] (-2.9,-3) -- (-3.0,-3);
\node[] (a)at(-3.2,-1.5){$L_{3}: $};
\draw[>=latex,black!70,dashed,very thick] (-1.5,-1.5) -- (-3.0,-1.5);
\node[] (a)at(-3.2,0){$L_{4}: $};
\draw[>=latex,black!70,dashed,very thick] (5.65,0) -- (-3.0,0);

\node[draw,rectangle,rounded corners=3pt,vert] (a)at(7.5,0){\color{black}$\{1\},\{2\},\{3\},\{4\}$};
\node[draw,rectangle,rounded corners=3pt,black!60] (b)at(0,-1.5){$\{1\},\{2\},\{3,4\}$};
\node[draw,rectangle,rounded corners=3pt,black!60] (c)at(3,-1.5){$\{3\},\{4\},\{1,2\}$};
\node[draw,rectangle,rounded corners=3pt,black!60] (d)at(6,-1.5){$\{1\},\{3\},\{2,4\}$};
\node[draw,rectangle,rounded corners=3pt,black!60] (e)at(9,-1.5){$\{2\},\{4\},\{1,3\}$};
\node[draw,rectangle,rounded corners=3pt,red] (f)at(12,-1.5){\color{black}$\{1\},\{4\},\{2,3\}$};
\node[draw,rectangle,rounded corners=3pt,vert] (g)at(15,-1.5){\color{black}$\{2\},\{3\},\{1,4\}$};

\node[draw,rectangle,rounded corners=3pt,red] (h)at(-1.5,-3){\color{black}$\{1\},\{2,3,4\}$};
\node[draw,rectangle,rounded corners=3pt,vert] (i)at(1.5,-3){\color{black}$\{1,2\},\{3,4\}$};
\node[draw,rectangle,rounded corners=3pt,vert] (j)at(4.5,-3){\color{black}$\{2\},\{1,3,4\}$};
\node[draw,rectangle,rounded corners=3pt,vert] (k)at(7.5,-3){\color{black}$\{1,3\},\{2,4\}$};
\node[draw,rectangle,rounded corners=3pt,vert] (l)at(10.5,-3){\color{black}$\{3\},\{1,2,4\}$};
\node[draw,rectangle,rounded corners=3pt,red] (m)at(13.5,-3){\color{black}$\{1,4\},\{2,3\}$};
\node[draw,rectangle,rounded corners=3pt,vert] (n)at(16.5,-3){\color{black}$\{4\},\{1,2,3\}$};
\node[draw,rectangle,rounded corners=3pt,red] (o)at(7.5,-4.5){\color{black}$\{1,2,3,4\}$};

\draw[-,>=latex,black!60,dashed] (0,-1.17) -- (7.5,-0.32);
\draw[-,>=latex,black!60,dashed] (3,-1.17) -- (7.5,-0.32);
\draw[-,>=latex,black!60,dashed] (6,-1.17) -- (7.5,-0.32);
\draw[-,>=latex,black!60,dashed] (9,-1.17) -- (7.5,-0.32);
\draw[->,>=latex,vert,dashed] (12,-1.17) -- (7.5,-0.32);
\draw[-,>=latex,black!60,dashed] (15,-1.17) -- (7.5,-0.32);

\draw[-,>=latex,black!60,dashed] (-1.5,-2.65) -- (0,-1.83);
\draw[-,>=latex,black!60,dashed] (-1.5,-2.65) -- (6,-1.83);
\draw[<-,>=latex,blue,dashed,very thick] (-1.05,-2.65) -- (12,-1.83);
\draw[-,>=latex,black!60,thick,dashed] (1.5,-2.65) -- (0,-1.83);
\draw[-,>=latex,black!60,thick,dashed] (1.5,-2.65) -- (3,-1.83);
\draw[-,>=latex,black!60,dashed] (4.5,-2.65) -- (0,-1.83);
\draw[-,>=latex,black!60,dashed] (4.5,-2.65) -- (9,-1.83);
\draw[-,>=latex,black!60,dashed] (4.5,-2.65) -- (15,-1.83);
\draw[-,>=latex,black!60,thick,dashed] (7.5,-2.65) -- (6,-1.83);
\draw[-,>=latex,black!60,thick,dashed] (7.5,-2.65) -- (9,-1.83);
\draw[-,>=latex,black!60,dashed] (10.5,-2.65) -- (3,-1.83);
\draw[-,>=latex,black!60,dashed] (10.5,-2.65) -- (6,-1.83);
\draw[-,>=latex,black!60,dashed] (10.5,-2.65) -- (15,-1.83);
\draw[->,>=latex,red,thick] (13.5,-2.65) -- (12,-1.83);
\draw[->,>=latex,vert,thick,dashed] (13.5,-2.65) -- (15,-1.83);
\draw[-,>=latex,black!60,dashed] (16.5,-2.65) -- (3,-1.83);
\draw[-,>=latex,black!60,dashed] (16.5,-2.65) -- (9,-1.83);
\draw[<-,>=latex,blue,dashed,very thick] (16.85,-2.65) -- (12,-1.83);

\draw[->,>=latex,red] (7.5,-4.17) -- (-1.5,-3.35);
\draw[->,>=latex,vert,dashed] (7.5,-4.17) -- (1.5,-3.35);
\draw[->,>=latex,vert,dashed] (7.5,-4.17) -- (4.5,-3.35);
\draw[->,>=latex,vert,dashed] (7.5,-4.17) -- (7.5,-3.35);
\draw[->,>=latex,vert,dashed] (7.5,-4.17) -- (10.5,-3.35);
\draw[->,>=latex,red] (7.5,-4.17) -- (13.5,-3.35);
\draw[->,>=latex,vert,dashed] (7.5,-4.17) -- (16.5,-3.35);

\draw[->,>=latex,black] (12.4,-4.3) -- (10.5,-3.8);
\node[red] (o)at(12.5,-4.4){\normalsize \color{black}Iteration $1$};

\draw[->,>=latex,black] (14.56,-0.82) -- (12.788,-2.2);
\node[red] (o)at(14.77,-0.7){\normalsize \color{black}Iteration $2$};

\draw[->,>=latex,black] (-1.43,-4.15) -- (1.4,-3.7);
\node[red] (o)at(-2.15,-4.28){\normalsize \color{black}Iteration $3$};

            \end{tikzpicture}%
        }
\caption{An illustration of the three phases of our algorithm. The numbers represent the agents. For example, $1$ represents agent $a_{1}$. In the first iteration, the start node is the bottom node. The child nodes are represented by the nodes that are directly connected to the bottom node. All of these child nodes are candidates to become the new start node. However, the coalition structure $\{\{a_{1},a_{4}\},\{a_{2},a_{3}\}\}$ in level 2 has, we assume, the highest value and hence it is selected 
and becomes 
the new start node. 
The newly generated child nodes of this coalition structure are $\{\{a_{1}\},\{a_{4}\},\{a_{2},a_{3}\}\}$ and $\{\{a_{2}\},\{a_{3}\},\{a_{1},a_{4}\}\}$. We assume that $\{\{a_{1}\},\{a_{4}\},\{a_{2},a_{3}\}\}$ is the highest valued coalition structure between the child nodes and hence it becomes the new start node. For the third iteration, there are child nodes that can be generated by merging coalitions through the blue edges, which are $\{\{a_{1}\},\{a_{2},a_{3},a_{4}\}\}$ and $\{\{a_{4}\},\{a_{1},a_{2},a_{3}\}\}$, but these are already generated. Hence, the only new child node generated is $\{\{a_{1}\},\{a_{2}\},\{a_{3}\},\{a_{4}\}\}$, which results from splitting a coalition. The algorithm then selects the coalition structure $\{\{a_{1}\},\{a_{2},a_{3},a_{4}\}\}$, which is, we assume, the highest valued one, and so on. Of course, after each selection, the best coalition structure is updated.} \label{CSGSALDAE}
\vspace{15pt}
\end{figure}

While these steps provide the foundational framework of the algorithm, additional optimizations and strategies are subsequently introduced to enhance performance, offering further refinements beyond the initial three-step process.

\subsection{Exploring Solutions through Bridging Paths for Enhanced Solution Quality}

The algorithm maintains a list of nodes and expands them based on the 3 steps. Hence, the algorithm moves from one node to another, seeking to improve the solution. 
In addition to these steps, each time a better solution is found, the algorithm creates a path of nodes between the previous best solution $\mathcal{S}_{l}$ and the new one $\mathcal{S}_{n}$. This is done to explore the possibility of finding even better solutions along the path between the two solutions. The motivation for this is that the distribution of coalitions and agents in the two current best solutions is related to the quality of the solutions. Hence, a slight change in these coalition structures could lead to even better solutions.

To illustrate this, consider a search graph with 8 agents. If the two current best solutions are $\mathcal{S}_{l} = \{\{a_{2},a_{7}\},\{a_{1},a_{4}\},\{a_{3},a_{5},a_{6},a_{8}\}\}$ (located on level 3) and $\mathcal{S}_{n} = \{\{a_{2}\},\{a_{3}\},\{a_{7}\},\{a_{1},a_{4}\},\{a_{5},a_{6},a_{8}\}\}$ (located on level~5), we can split the first coalition of $\mathcal{S}_{l}$ to obtain the coalition structure $\mathcal{CS}_{1} = \{\{a_{2}\},\{a_{7}\},\{a_{1},a_{4}\},\{a_{3},a_{5},a_{6},a_{8}\}\}$. Then, we can also split the fourth coalition $\{a_{3},a_{5},a_{6},a_{8}\}$ of $\mathcal{CS}_{1}$ into two coalitions $\{a_{3}\}$ and $\{a_{5},a_{6},a_{8}\}$ to obtain the coalition structure $\mathcal{S}_{n}$. These two splits create a path between the two current best solutions. The solutions along this path ($\mathcal{CS}_{1}$), which may not have been processed yet, may be better than both $\mathcal{S}_{l}$ and $\mathcal{S}_{n}$. 
It is worth noting that the path may contain multiple solutions, depending on the distance between $\mathcal{S}_{l}$ and $\mathcal{S}_{n}$. However, it is not always possible to construct a path between two solutions by only splitting or merging coalitions. In some cases, agents may be in completely different coalitions than in the last best solution. Therefore, a path between the two current best solutions may be found by combining splits and merges of coalitions.

In what follows, we refer to the node that contains the grand coalition as the bottom node and the node that contains the singleton coalitions as the top node. The following properties hold.

\begin{observation}
Given $n$ agents, let $\mathcal{N}$ be a node of level $l$. Then, it holds that:

\begin{itemize}
    \item The bottom node can be reached from $\mathcal{N}$ with $l-1$ merges and $\mathcal{N}$ can be reached from the bottom node with $l-1$ splits. 
    \item The top node can be reached from $\mathcal{N}$ with $n-l$ splits and $\mathcal{N}$ can be reached from the top node with $n-l$ merges. 
\end{itemize}
\end{observation}

Recall that a path goes through several coalition structures by splitting and merging coalitions. We refer to the number of splits and merges to reach one coalition structure from another by the size of the path. 

\begin{observation}
Let $\mathcal{CS}_{i}$ and $\mathcal{CS}_{j}$ be two coalition structures, where $\mathcal{CS}_{i} \neq \mathcal{CS}_{j}$. Then, there is always a path between $\mathcal{CS}_{i}$ and $\mathcal{CS}_{j}$ of size at most $n-1$.
\end{observation}

The proofs of Observation 1 and Observation 2 are given in the appendix. Given the previous best solution $\mathcal{CS}_{l}$ and the new best solution $\mathcal{CS}_{n}$, to reach the coalition structure $\mathcal{CS}_{n}$ from $\mathcal{CS}_{l}$, we propose in this paper three strategies: SPLIT-THEN-MERGE, MERGE-THEN-SPLIT, and APPROACH-THEN-SWAP.

\subsubsection{SPLIT-THEN-MERGE}

One alternative to reach $\mathcal{CS}_{n}$ from $\mathcal{CS}_{l}$ is based on Observations~1 and~2. Starting from $\mathcal{CS}_{l}$, the algorithm splits the coalitions one by one until reaching the top node. Once the top node is reached, the algorithm merges the coalitions until the desired node $\mathcal{CS}_{n}$ is reached. To ensure that $\mathcal{CS}_{n}$ is reached, the algorithm avoids merging coalitions whose agents are not part of the same coalition in $\mathcal{CS}_{n}$. 
An illustration of this strategy can be seen in Figure \ref{threeStrategies} as the orange path.

\subsubsection{MERGE-THEN-SPLIT}

This strategy involves a descending phase followed by an ascending phase. 
During the descending phase, the algorithm starts from the last best solution $\mathcal{CS}_{l}$ and merges coalitions to form the grand coalition at the bottom node. Then, during the ascending phase, the algorithm splits the coalitions to reach the target solution $\mathcal{CS}_{n}$. 
During this phase, the algorithm does not separate agents that are in the same coalition in $\mathcal{CS}_{n}$, as this would prevent the algorithm from reaching its target solution. An illustration of this strategy can be seen in Figure \ref{threeStrategies}, where it is represented by the purple path.

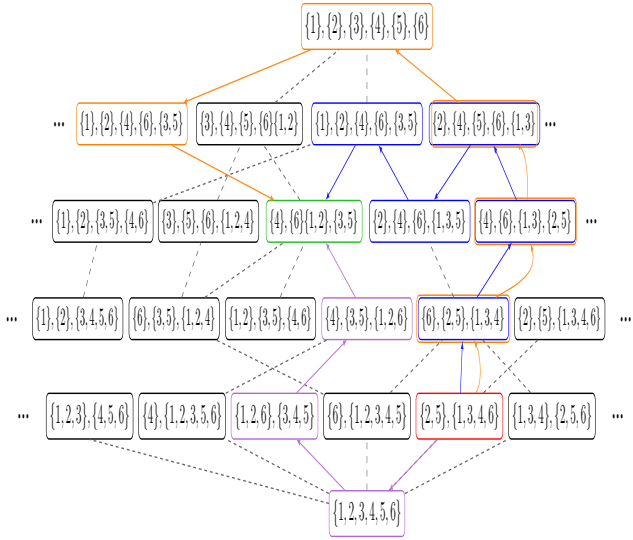
\begin{figure}[t]
\centering
\normalsize
\resizebox{1.0\columnwidth}{7.1cm}{%
\begin{tikzpicture}
\Large
\node[draw,rectangle,rounded corners=3pt,orange] (a0)at(7.5,0){\color{black}$\{1\},\{2\},\{3\},\{4\},\{5\},\{6\}$};

\large
\node[] (k)at(-4.15,-1.5){\huge \color{black}...};
\node[draw,rectangle,rounded corners=3pt,orange] (b1)at(-1.4,-1.5){\color{black}$\{1\},\{2\},\{4\},\{6\},\{3,5\}$};
\node[draw,rectangle,rounded corners=3pt,black] (c1)at(3.05,-1.5){\color{black}$\{3\},\{4\},\{5\},\{6\}\{1,2\}$};
\node[draw,rectangle,rounded corners=3pt,blue] (d1)at(7.5,-1.5){\color{black}$\{1\},\{2\},\{4\},\{6\},\{3,5\}$};
\node[draw,rectangle,rounded corners=3pt,blue] (e1)at(11.95,-1.5){\color{black}$\{2\},\{4\},\{5\},\{6\},\{1,3\}$};
\node[draw,rectangle,rounded corners=3pt,orange,minimum height=0.712cm,minimum width=3.99cm] (xx1)at(11.95,-1.5){};
\node[] (k)at(14.45,-1.5){\huge \color{black}...};

\large
\node[] (k)at(-5.0,-3){\huge \color{black}...};
\node[draw,rectangle,rounded corners=3pt,black] (a2)at(-2.5,-3){\color{black}$\{1\},\{2\},\{3,5\},\{4,6\}$};
\node[draw,rectangle,rounded corners=3pt,black] (b2)at(1.5,-3){\color{black}$\{3\},\{5\},\{6\},\{1,2,4\}$};
\node[draw,rectangle,rounded corners=3pt,vert] (c2)at(5.5,-3){\color{black}$\{4\},\{6\}\{1,2\},\{3,5\}$};
\node[draw,rectangle,rounded corners=3pt,blue] (d2)at(9.5,-3){\color{black}$\{2\},\{4\},\{6\},\{1,3,5\}$};
\node[draw,rectangle,rounded corners=3pt,blue] (e2)at(13.5,-3){\color{black}$\{4\},\{6\},\{1,3\},\{2,5\}$};
\node[draw,rectangle,rounded corners=3pt,orange,minimum height=0.712cm,minimum width=3.862cm] (xx2)at(13.5,-3){};
\node[] (k)at(16.0,-3){\huge \color{black}...};

\large
\node[] (k)at(-5.95,-4.5){\huge \color{black}...};
\node[draw,rectangle,rounded corners=3pt,black] (a3)at(-3.45,-4.5){\color{black}$\{1\},\{2\},\{3,4,5,6\}$};
\node[draw,rectangle,rounded corners=3pt,black] (b3)at(0.2,-4.5){\color{black}$\{6\},\{3,5\},\{1,2,4\}$};
\node[draw,rectangle,rounded corners=3pt,black] (c3)at(3.85,-4.5){\color{black}$\{1,2\},\{3,5\},\{4,6\}$};
\node[draw,rectangle,rounded corners=3pt,lilas] (d3)at(7.5,-4.5){\color{black}$\{4\},\{3,5\},\{1,2,6\}$};
\node[draw,rectangle,rounded corners=3pt,blue] (e3)at(11.15,-4.5){\color{black}$\{6\},\{2,5\},\{1,3,4\}$};
\node[draw,rectangle,rounded corners=3pt,orange,minimum height=0.725cm,minimum width=3.53cm] (xx3)at(11.15,-4.5){};
\node[draw,rectangle,rounded corners=3pt,black] (f3)at(14.8,-4.5){\color{black}$\{2\},\{5\},\{1,3,4,6\}$};
\node[] (k)at(17.3,-4.5){\huge \color{black}...};

\Large
\node[] (k)at(-5.5,-6){\huge \color{black}...};
\node[draw,rectangle,rounded corners=3pt,black] (a4)at(-3.0,-6){\color{black}$\{1,2,3\},\{4,5,6\}$};
\node[draw,rectangle,rounded corners=3pt,black] (b4)at(0.5,-6){\color{black}$\{4\},\{1,2,3,5,6\}$};
\node[draw,rectangle,rounded corners=3pt,lilas] (c4)at(4.0,-6){\color{black}$\{1,2,6\},\{3,4,5\}$};
\node[draw,rectangle,rounded corners=3pt,black] (d4)at(7.5,-6){\color{black}$\{6\},\{1,2,3,4,5\}$};
\node[draw,rectangle,rounded corners=3pt,red] (e4)at(11.0,-6){\color{black}$\{2,5\},\{1,3,4,6\}$};
\node[draw,rectangle,rounded corners=3pt,black] (f4)at(14.5,-6){\color{black}$\{1,3,4\},\{2,5,6\}$};
\node[] (k)at(17.0,-6){\huge \color{black}...};

\Large
\node[draw,rectangle,rounded corners=3pt,lilas] (a5)at(7.5,-7.5){\color{black}$\{1,2,3,4,5,6\}$};

\draw[->,>=latex,blue!90] (e4) -- (e3);
\draw[->,>=latex,blue!90] (e3) -- (e2);
\draw[->,>=latex,blue!90] (e2) -- (e1);

\draw[-,>=latex,black!60,dashed] (e4) -- (a5);

\draw[-,>=latex,black!60,dashed] (a5) -- (a4.south);
\draw[-,>=latex,black!60,dashed] (a5) -- (b4);
\draw[-,>=latex,black!60,dashed] (a5) -- (d4);
\draw[-,>=latex,black!60,dashed] (a5) -- (f4);

\draw[-,>=latex,black!60,dashed] (b4) -- (d3);
\draw[-,>=latex,black!60,dashed] (d4) -- (b3);
\draw[-,>=latex,black!60,dashed] (d4) -- (e3);
\draw[-,>=latex,black!60,dashed] (e4) -- (f3);
\draw[-,>=latex,black!60,dashed] (f4) -- (e3);

\draw[-,>=latex,black!60,dashed] (a3) -- (a2);
\draw[-,>=latex,black!60,dashed] (b3) -- (b2);
\draw[-,>=latex,black!60,dashed] (b3) -- (c2);
\draw[-,>=latex,black!60,dashed] (c3) -- (c2);
\draw[-,>=latex,black!60,dashed] (e3) -- (d2);

\draw[-,>=latex,black!60,dashed] (a2) -- (d1);
\draw[-,>=latex,black!60,dashed] (b2) -- (c1);
\draw[-,>=latex,black!60,dashed] (c2) -- (c1);
\draw[->,>=latex,blue!90] (d1) -- (c2);
\draw[->,>=latex,blue!90] (d2) -- (d1);
\draw[->,>=latex,blue!90] (e1) -- (d2);

\draw[-,>=latex,black!60,dashed] (c1) -- (a0);
\draw[-,>=latex,black!60,dashed] (d1) -- (a0);

\draw[->,>=latex,orange!90] (e4) to[out=30,in=-40]  (e3);
\draw[->,>=latex,orange!90] (e3) to[out=15,in=-60] (e2);
\draw[->,>=latex,orange!90] (e2) to[out=80,in=-15] (e1);
\draw[->,>=latex,orange!90] (e1) -- (a0);
\draw[->,>=latex,orange!90] (a0) -- (b1);
\draw[->,>=latex,orange!90] (b1) -- (c2);

\draw[->,>=latex,lilas!90] (e4) -- (a5);
\draw[->,>=latex,lilas!90] (a5) -- (c4);
\draw[->,>=latex,lilas!90] (c4) -- (d3);
\draw[->,>=latex,lilas!90] (d3) -- (c2);


            \end{tikzpicture}%
        }
\caption{An illustration of the 3 strategies with a partial graph of 6 agents. The numbers represent the agents. For example, $1$ represents agent $a_{1}$. The last best solution is the red node and the new best one is the green node. The nodes of the path between the two current best solutions of the SPLIT-THEN-MERGE, MERGE-THEN-SPLIT and APPROACH-THEN-SWAP strategies are represented by the orange, purple, and blue nodes, respectively.} \label{threeStrategies}
\end{figure}

\subsubsection{APPROACH-THEN-SWAP}

This strategy consists of two phases: an approach phase and a swap phase. In the approach phase, the algorithm generates an intermediate solution that is at the same level as $\mathcal{CS}_{n}$. 
To reach this intermediate solution, the algorithm performs a series of splits or a series of merges on the coalitions of $\mathcal{CS}_{l}$. If the level of $\mathcal{CS}_{n}$ is higher than that of $\mathcal{CS}_{l}$, the algorithm performs $l_{2} - l_{1}$ splits, where $l_{1}$ and $l_{2}$ are the levels of $\mathcal{CS}_{l}$ and $\mathcal{CS}_{n}$, respectively. If the level of $\mathcal{CS}_{n}$ is lower than that of $\mathcal{CS}_{l}$, the algorithm performs $l_{1} - l_{2}$ merges. If the two coalition structures are at the same level, no splitting or merging is performed.

Once the intermediate solution $\mathcal{CS}_{i}$ that has the same number of coalitions and the same number of agents in each coalition as $\mathcal{CS}_{n}$ is reached, the swap phase begins. This phase involves swapping agents between coalitions in order to reach $\mathcal{CS}_{n}$. 
Given a coalition structure $\mathcal{CS} =\{\mathcal{C}_{1}, \mathcal{C}_{2},...,\mathcal{C}_{k}\}$, swapping an agent $a_{i}$ of a coalition $\mathcal{C}_{i}$ with an agent $a_{j}$ of a coalition $\mathcal{C}_{j}$ 
results in a new coalition structure $(\mathcal{CS} \setminus \{\mathcal{C}_{i}, \mathcal{C}_{j}\}) \bigcup \{\mathcal{C}'_{i},\mathcal{C}'_{j}\}$, where $\mathcal{C}'_{i} = (\mathcal{C}_{i}\setminus\{a_{i}\}) \bigcup \{a_{j}\}$ and $\mathcal{C}'_{j} = (\mathcal{C}_{j}\setminus\{a_{j}\}) \bigcup \{a_{i}\}$. This rule is equivalent to two series of split and merge. 
To perform a swap, the algorithm splits $\mathcal{C}_{i}$ into $\mathcal{C}_{i} \setminus \{a_{i}\}$ and $\{a_{i}\}$. Then, merges $\{a_{i}\}$ with $\mathcal{C}_{j}$ to obtain $\mathcal{C}_{j} \bigcup \{a_{i}\}$. Again, splits $\mathcal{C}_{j} \bigcup \{a_{i}\}$ into $(\mathcal{C}_{j} \bigcup \{a_{i}\}) \setminus \{a_{j}\}$ and $\{a_{j}\}$. Finally, merges $\mathcal{C}_{i} \setminus \{a_{i}\}$ with $\{a_{j}\}$ to obtain $(\mathcal{C}_{i} \setminus \{a_{i}\}) \bigcup \{a_{j}\}$. For example, swapping $a_{2}$ with $a_{3}$ in the coalition structure $\{\{a_{4}\},\{a_{1},a_{2}\},\{a_{3},a_{5}\}\}$ is done as shown in Figure \ref{approachthenswap}.

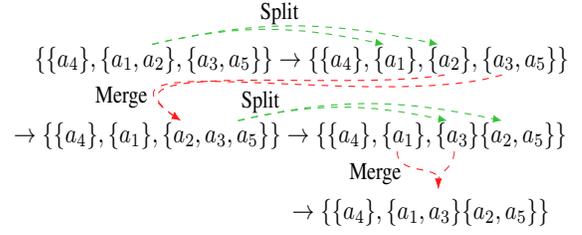
\begin{figure}[h!]
\centering
\normalsize
\resizebox{0.9\columnwidth}{3.2cm}{%
\begin{tikzpicture}

\node[] (a)at(-0.35,0.6){\small Split};

\node[] (e3)at(0.0,0.0){\color{black}$\{\{a_{4}\},\{a_{1},a_{2}\},\{a_{3},a_{5}\}\} \rightarrow \{\{a_{4}\},\{a_{1}\},\{a_{2}\},\{a_{3},a_{5}\}\}$};

\draw[->,>=latex,vert!90,dashed] (-2.35,0.2) to[out=11,in=170] (2.15,0.2);
\draw[->,>=latex,vert!90,dashed] (-2.35,0.2) to[out=11,in=170] (1.30,0.2);

\draw[->,>=latex,red!90,dashed] (3.15,-0.2) to[out=194,in=150] (-1.90,-0.8);
\draw[->,>=latex,red!90,dashed] (2.25,-0.2) to[out=188,in=150] (-1.90,-0.8);

\node[] (a)at(-2.83,-0.49){\small Merge};

\node[] (e3)at(-0.18,-1.0){\color{black}$\rightarrow \{\{a_{4}\},\{a_{1}\},\{a_{2},a_{3},a_{5}\}\} \rightarrow \{\{a_{4}\},\{a_{1}\},\{a_{3}\}\{a_{2},a_{5}\}\}$};

\node[] (a)at(-0.65,-0.56){\small Split};

\draw[->,>=latex,vert!90,dashed] (-1.0,-0.8) to[out=11,in=170] (2.29,-0.8);
\draw[->,>=latex,vert!90,dashed] (-1.0,-0.8) to[out=11,in=170] (3.15,-0.8);

\draw[->,>=latex,red!90,dashed] (1.5,-1.2) to[out=270,in=90] (2.15,-1.7);
\draw[->,>=latex,red!90,dashed] (2.4,-1.2) to[out=270,in=90] (2.15,-1.7);

\node[] (a)at(1.16,-1.49){\small Merge};

\node[] (e3)at(1.87,-2.0){\color{black}$\rightarrow \{\{a_{4}\},\{a_{1},a_{3}\}\{a_{2},a_{5}\}\}$};

\end{tikzpicture}%
        }
\caption{An illustration of the APPROACH-THEN-SWAP strategy.} \label{approachthenswap}
\end{figure}

All the intermediate coalition structures that the path generates are evaluated with the aim to find better solutions. Figure \ref{threeStrategies} shows an illustration of this strategy (see the blue path).




\subsection{Optimizing the Search for Global Optima through Memory Management}

SALDAE maintains three lists of nodes in memory: OPEN, SUBSTITUTE, and RESERVE. OPEN is sorted in descending order according to the values of the coalition structures represented by the nodes it contains. It is constructed using the three steps of the search process (see Section~4). 
SUBSTITUTE contains the nodes visited by the algorithm during the path search between the current best solutions. These nodes are not expanded unless they improve the current solution. Two hyperparameters, $\theta$ and $\omega$, are introduced to control the memory usage of the algorithm. $\theta$ determines the number of child nodes to keep in memory for each expanded node, while $\omega$ is the lower bound on the minimum value required for a node to be maintained in the list OPEN. Let $\mathcal{CS}^{+}$ be the current best solution. The list OPEN contains nodes $\mathcal{CS}$ such that $v(\mathcal{CS}) \geq \omega \times v(\mathcal{CS}^{+})$, while RESERVE contains nodes $\mathcal{CS}$ with $v(\mathcal{CS}) < \omega \times v(\mathcal{CS}^{+})$. The list OPEN is used by SALDAE for the search and is replaced by either the SUBSTITUTE or RESERVE list if it becomes empty.

\subsubsection{Replacing OPEN with SUBSTITUTE}

Since the list OPEN contains the nodes where $v(\mathcal{CS}) \geq \omega \times v(\mathcal{CS}^{+})$, it is possible that after several iterations OPEN becomes empty. This can happen if the nodes expanded by SALDAE do not generate child nodes that meet the condition ($v(\mathcal{CS}) \geq \omega \times v(\mathcal{CS}^{+})$), i.e. the child nodes are directly added to the list RESERVE. 
In this case, the algorithm replaces OPEN with SUBSTITUTE and continues the search. This allows SALDAE to explore different areas of the solution space that may be more promising, instead of staying in a part of the space that has not produced better solutions or at least solutions that would have remained in OPEN.

\subsubsection{Replacing OPEN with RESERVE}

In SALDAE, SUBSTITUTE is constructed after finding a new best solution. However, when OPEN is replaced with SUBSTITUTE, the list SUBSTITUTE remains empty, until a better solution is found. If before finding a better solution OPEN becomes empty, SALDAE replaces OPEN with the list RESERVE and continues the search. 
This way, SALDAE prioritizes searching in parts of the solution space that are more promising, but can also return to the list RESERVE if it does not find more promising parts.

By using these three lists, SALDAE can effectively guide its search to focus on more promising areas of the solution space, while still being able to revisit areas that may have been overlooked.

\subsection{Multiagent Search}

To speed up finding high-quality solutions, SALDAE employs several search agents $\{s_{1},...,s_{m}\}$. Each agent $s_{i}$ has a start node and maintains three lists of nodes OPEN$_{i}$, RESERVE$_{i}$, and SUBSTITUTE$_{i}$. SALDAE affects one agent for each of the bottom and top nodes, and random nodes to other agents. Then, each agent moves from a node to an adjacent node searching for better solutions.

When a search agent $s_{i}$ expands a node and generates child nodes, it checks for conflicts among the lists of the search agents. We say that two agents have a conflict iff they consider the same coalition structure for evaluation. If there are no conflicts, the agent $s_{i}$ adds the generated child node to the list OPEN$_{i}$ or RESERVE$_{i}$. 
Otherwise, $s_{i}$ resolves the conflict by selecting the search agent that will keep the generated child node, using the following techniques. We introduce the following techniques to resolve the conflicts.

\subsubsection{Bypassing Conflicts}

If a child node is found in an OPEN$_{j}$ or RESERVE$_{j}$ list of another search agent $s_{j}$, the search agent $s_{i}$ is prohibited from using the conflicting child node and $s_{j}$ keeps it. This ensures that the first search agent to consider the node is the one that keeps it. This 
is different from the bypassing conflicts in CBS, where SALDAE aims to bypass the conflicts by generating coalition structures that avoid leading to conflicts in the first place. The goal of doing this in SALDAE is to guarantee that the search agents do not search the same coalition structures at the same time and explore different areas of the solution space. 

\subsubsection{Managing Conflicts}

Since the OPEN and RESERVE lists are not the same for each agent, the search agent $s_{i}$ compares 
the ranking of the child node in its own list with the ranking in the list of the conflicting search agent $s_{j}$. 
If the search agent $s_{i}$ allows the expansion of the conflicting child node first, it will remove the node from the list of the other search agent $s_{j}$ and add it to its own list. 
This indicates that the treatment of the conflicting child node is removed from the search agent $s_{j}$. Otherwise, the search agent $s_{i}$ is prohibited from using the conflicting child node and $s_{j}$ keeps it. 
This ensures that the search agent that allows the expansion of the conflicting child node first gets to keep it.

Algorithm \ref{SALDAE} shows the pseudocode of SALDAE. Lines 8 and 12 are discussed in the section 5 of the paper. SALDAE runs in parallel $m$ search agents $s_{i}$. ExecutePathStrategy function constructs a path between the last best solutions and adds the visited nodes to the list SUBSTITUTE$_{i}$. AddChildNodesToOPENorRESERVE function distributes the nodes between the lists OPEN$_{i}$ and RESERVE$_{i}$ according to their coalition structure values:

\begin{enumerate}
    \item if $v(\mathcal{N}) \geq \omega \times v(\mathcal{CS}^{+})$ then add $\mathcal{N}$ to OPEN$_{i}$;
    \item else add $\mathcal{N}$ to RESERVE$_{i}$.
\end{enumerate}

\begin{algorithm}[t]
\KwIn{A number of agents $n$ and the values $v(\mathcal{C})$ of the coalitions $\mathcal{C}$. 
}
\KwOut{The highest-valued coalition structure found $CS^{*}$ and its value.}
\DontPrintSemicolon
Generate root node R\;
$\mathcal{CS}^{*} \leftarrow $ R\;
$\mathcal{V}^{*} \leftarrow v($R$)$\;
Generate start nodes R$_{i}$ for the search agents $s_{i}$\;
 \nonl $\triangleright$ Begin parallel\;
 \nonl $\triangleright$ SALDAE runs in parallel $m$ search agents $s_{i}$\;
\If{$v($R$_{i}) > \mathcal{V}^{*}$}{
$\mathcal{CS}^{*} \leftarrow $ R$_{i}$\;
$\mathcal{V}^{*} \leftarrow v($R$_{i})$\;
}
Childs $\leftarrow$ ComputeChildNodes(R$_{i}$)\;
AddChildNodesToOPENorRESERVE(Childs)\;
\While{OPEN$_{i}$ is not empty}{
$\mathcal{N} \leftarrow$ OPEN$_{i}$.pop() $\triangleright$ \textit{this selects the highest-valued node}\;
Childs $\leftarrow$ ComputeChildNodes($\mathcal{N}$)\;
AddChildNodesToOPENorRESERVE(Childs)\;
\If{$v(\mathcal{N}) > \mathcal{V}^{*}$}{

ExecutePathStrategy($\mathcal{CS}^{*}$,$\mathcal{N}$)\;

\If{$v(\mathcal{N}) > \mathcal{V}^{*}$}{
$\mathcal{CS}^{*} \leftarrow \mathcal{N}$\;
$\mathcal{V}^{*} \leftarrow v(\mathcal{N})$\;
}
}
RemoveNodeFromOPEN$_{i}$($\mathcal{N}$)\;
\If{OPEN$_{i}$ is empty}{
\If{SUBSTITUTE$_{i}$ is empty}{
Replace OPEN$_{i}$ by RESERVE$_{i}$\;
}
\Else{
Replace OPEN$_{i}$ by SUBSTITUTE$_{i}$\;
}
}
}
\nonl $\triangleright$ End parallel\;

Return $\mathcal{CS}^{*}$, $\mathcal{V}^{*}$\;
 \caption{SALDAE algorithm} \label{SALDAE}
\end{algorithm}

\section{Selecting Child Nodes}

Due to the large number of possible child nodes in the coalition structure graph, it is infeasible to generate all of them, especially for large numbers of agents, this makes the search space intractable and highlights the importance of using more efficient selection strategies. Here, we present and introduce two child node selection methods.

\subsection{Quantity-Based Selection}

A straightforward idea for generating child nodes that can reduce computational burden and memory requirements is to select a number of child nodes and then keep the best ones in OPEN$_{i}$ and RESERVE$_{i}$. 
Let $\mathcal{N}_{c}$ and $\mathcal{N}_{a}$ be the number of child nodes to generate for the list OPEN$_{i}$ and the number of child nodes to actually add to the list OPEN$_{i}$, with $\mathcal{N}_{a} < \mathcal{N}_{c}$. 
In the first step of generating child nodes, we generate $\mathcal{N}_{c}$ child nodes, when they exist, whose values are greater than or equal to a threshold value, $\omega \times v(\mathcal{CS}^{+})$. From these generated child nodes, we select the $\mathcal{N}_{a}$ nodes with the highest values and add them to the list OPEN$_{i}$. Any generated nodes that have values less than $\omega \times v(\mathcal{CS}^{+})$ are added to the list RESERVE$_{i}$.

\begin{figure}[t]
\begin{center}
\small
\resizebox{0.48\columnwidth}{3.9cm}{%
     \begin{subfigure}{0.25\textwidth}
         \centering
         \begin{tikzpicture}
	\begin{axis}[	grid= major ,
            title=(a) Normal ,
			width=0.8\textwidth ,
            xlabel = {Number of agents} ,
			ylabel = {Solution quality (\%)} ,
            width=5cm,height=5cm,
            xtick={4,6,8,10,12,14,16,18,20},
            xticklabels={4,6,8,10,12,14,16,18,20},
            ytick={70,80,90,95,99,100},
            yticklabels={70,80,90,95,99,100},
            ymin=90,
            ymax=102,
            label style={font=\Large},
			title style={font=\Large},
			tick label style={font=\footnotesize},
            legend entries={SALDAE, CSG-UCT, PICS},
			legend style={at={(0,0.65)},anchor=north west,opacity=0.6,text opacity = 1}]
			\addplot+[only marks,color = blue,mark=square*,mark options={fill=cyan}] coordinates {(4,100) (5,100) (6,100) (7,100) (8,100) (9,100) (10,100) (11,100) (12,100) (13,100) (14,99.91) (15,99.97) (16,99.55) (17,99.47) (18,99.56) (19,99.44) (20,99.28) };//
			\addplot+[only marks,color = vert,mark=*,mark options={fill=vert}] coordinates {(4,100) (5,100) (6,100) (7,100) (8,99.7) (9,99.4) (10,99.2) (11,98.9) (12,98.8) (13,98.2) (14,97.81) (15,97.27) (16,96.23) (17,96.77) (18,96.46) (19,96.54) (20,96.31) };
            \addplot+[only marks,color = red,mark=triangle*,mark options={fill=red},ultra thin] coordinates {(4,100) (5,100) (6,100) (7,100) (8,100) (9,100) (10,100) (11,100) (12,100) (13,100) (14,99.9) (15,99.7) (16,99.2) (17,99.17) (18,98.76) (19,98.59) (20,98.68) };
            
	\end{axis}
    \end{tikzpicture}\\
         \label{fig:three sin x}
     \end{subfigure}
     }
     \hfill
     \resizebox{0.48\columnwidth}{3.9cm}{%
     \begin{subfigure}{0.25\textwidth}
         \centering
         \begin{tikzpicture}
	\begin{axis}[	grid= major ,
            title=(b) Zipf ,
			width=0.6\textwidth ,
            xlabel = {Number of agents} ,
            width=5cm,height=5cm,
            xtick={4,6,8,10,12,14,16,18,20},
            xticklabels={4,6,8,10,12,14,16,18,20},
            ytick={70,80,90,95,99,100},
            yticklabels={70,80,90,95,99,100},
            ymin=90,
            ymax=102,
            label style={font=\Large},
			title style={font=\Large},
			tick label style={font=\footnotesize},
            legend entries={SALDAE, CSG-UCT, PICS},
			legend style={at={(0,0.65)},anchor=north west,opacity=0.6,text opacity = 1}]
			\addplot+[only marks,color = blue,mark=square*,mark options={fill=cyan}] coordinates {(4,100) (5,100) (6,100) (7,100) (8,100) (9,100) (10,100) (11,100) (12,100) (13,100) (14,99.91) (15,99.97) (16,99.33) (17,99.27) (18,99.26) (19,99.14) (20,99.11) };//
			\addplot+[only marks,color = vert,mark=*,mark options={fill=vert}] coordinates {(4,100) (5,100) (6,100) (7,100) (8,99.7) (9,99.4) (10,99.2) (11,98.5) (12,98.3) (13,97.8) (14,97.31) (15,96.57) (16,96.03) (17,95.77) (18,95.26) (19,95.14) (20,95.11) };
            \addplot+[only marks,color = red,mark=triangle*,mark options={fill=red},ultra thin] coordinates {(4,100) (5,100) (6,100) (7,100) (8,100) (9,100) (10,99.8) (11,99.7) (12,99.3) (13,99.1) (14,99.31) (15,98.07) (16,97.93) (17,97.77) (18,97.26) (19,97.14) (20,97.11) };
	\end{axis}
    \end{tikzpicture}\\
         \label{fig:three sin x}
     \end{subfigure}
     }
     \hfill
     \resizebox{0.48\columnwidth}{3.9cm}{%
     \begin{subfigure}{0.25\textwidth}
         \centering
         \begin{tikzpicture}
	\begin{axis}[	grid= major ,
            title=(e) Gamma ,
			width=0.6\textwidth ,
            xlabel = {Number of agents} ,
			ylabel = {Solution quality (\%)} ,
            width=5cm,height=5cm,
            xtick={4,6,8,10,12,14,16,18,20},
            xticklabels={4,6,8,10,12,14,16,18,20},
            ytick={70,80,90,95,99,100},
            yticklabels={70,80,90,95,99,100},
            ymin=60,
            ymax=102,
            label style={font=\Large},
			title style={font=\Large},
			tick label style={font=\footnotesize},
            legend entries={SALDAE, CSG-UCT, PICS},
			legend style={at={(0,0.65)},anchor=north west,opacity=0.6,text opacity = 1}]
			\addplot+[only marks,color = blue,mark=square*,mark options={fill=cyan}] coordinates {(4,100) (5,100) (6,100) (7,100) (8,100) (9,100) (10,100) (11,100) (12,100) (13,100) (14,100) (15,99.99) (16,99.93) (17,99.97) (18,99.96) (19,99.94) (20,100) };//
			\addplot+[only marks,color = vert,mark=*,mark options={fill=vert}] coordinates {(4,100) (5,100) (6,100) (7,99.55) (8,97.1) (9,94.2) (10,91.96) (11,92.91) (12,89.86) (13,83.65) (14,84) (15,78.79) (16,73.63) (17,74.97) (18,71.16) (19,69.94) (20,65.1) };
            \addplot+[only marks,color = red,mark=triangle*,mark options={fill=red},ultra thin] coordinates {(4,100) (5,100) (6,100) (7,100) (8,100) (9,100) (10,99.96) (11,99.91) (12,99.86) (13,99.65) (14,99) (15,98.79) (16,97.63) (17,96.97) (18,96.16) (19,96.94) (20,96) };
	\end{axis}
    \end{tikzpicture}\\
         \label{fig:three sin x}
     \end{subfigure}
     }
     \hfill
     \resizebox{0.48\columnwidth}{3.9cm}{%
     \begin{subfigure}{0.25\textwidth}
         \centering
         \begin{tikzpicture}
	\begin{axis}[	grid= major ,
            title=(f) Exponential ,
			width=0.6\textwidth ,
            xlabel = {Number of agents} ,
            width=5cm,height=5cm,
            xtick={4,6,8,10,12,14,16,18,20},
            xticklabels={4,6,8,10,12,14,16,18,20},
            ytick={70,80,90,95,99,100},
            yticklabels={70,80,90,95,99,100},
            ymin=65,
            ymax=102,
            label style={font=\Large},
			title style={font=\Large},
			tick label style={font=\footnotesize},
            legend entries={SALDAE, CSG-UCT, PICS},
			legend style={at={(0,0.65)},anchor=north west,opacity=0.6,text opacity = 1}]
			\addplot+[only marks,color = blue,mark=square*,mark options={fill=cyan}] coordinates {(4,100) (5,100) (6,100) (7,100) (8,100) (9,100) (10,100) (11,100) (12,100) (13,99.8) (14,99.41) (15,99.18) (16,98.73) (17,98.87)
			(18,98.66) (19,96.64) (20,97.8) };//
			\addplot+[only marks,color = vert,mark=*,mark options={fill=vert}] coordinates {(4,100) (5,100) (6,100) (7,99.55) (8,97.1) (9,94.2) (10,91.96) (11,92.91) (12,89.86) (13,83.65) (14,84) (15,78.79) (16,73.63) (17,74.97) (18,71.16) (19,71.94) (20,70.8) };
            \addplot+[only marks,color = red,mark=triangle*,mark options={fill=red},ultra thin] coordinates {(4,100) (5,100) (6,100) (7,100) (8,99) (9,97) (10,94) (11,94) (12,93) (13,90) (14,89) (15,88.7) (16,87) (17,86.4) (18,86) (19,84.5) (20,84) };
	\end{axis}
    \end{tikzpicture}\\
         \label{fig:three sin x}
     \end{subfigure}
}

             \caption{Solution quality of SALDAE, PICS and CSG-UCT for sets of agents between 4 and 20. 
             }
        \label{solutionQuality}
\normalsize
\end{center}
\end{figure}
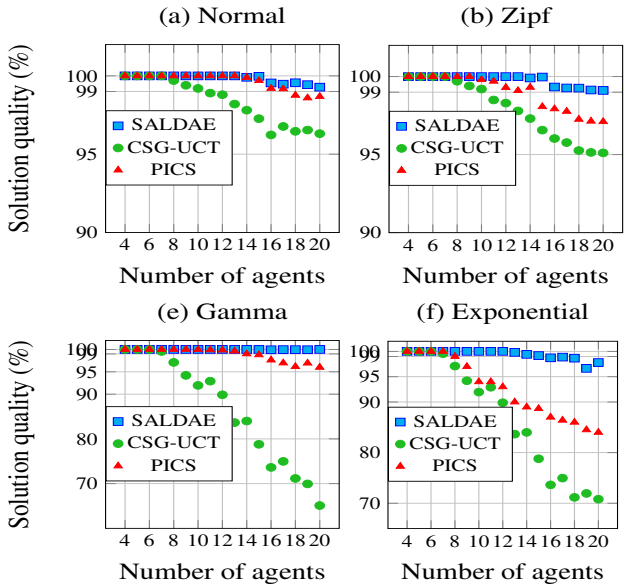

\subsection{Value-Based Selection}

The second idea is to select $\mathcal{N}_{a}$ child nodes to add to the list OPEN$_{i}$ and use an exit condition to stop the generation of child nodes based on the values of the child nodes. 
We start by fixing an initial threshold parameter $\gamma = v(\mathcal{CS}^{+})$, which allows us to specify the acceptance of the generated child nodes. 
At each iteration, we generate $\mathcal{N}_{a}$ random child nodes and add them to a temporary list. If one of the generated child nodes, $\mathcal{CS}$, has a value $v(\mathcal{CS}) > \gamma$, we add the $\mathcal{N}_{a}$ best nodes from the temporary list to OPEN$_{i}$ or RESERVE$_{i}$. Otherwise, we reduce the threshold of acceptance by half the difference between the current threshold and the best value of the generated child nodes, i.e., 
$\gamma = \gamma - (\frac{\gamma - v(\mathcal{CS}^{\#})}{2})$, where $v(\mathcal{CS}^{\#})$ is the value of the highest valued generated child node. This way, we give a better chance of satisfying the condition after each iteration. 
We then loop continuously until the exit condition is satisfied.

The pseudo-code of SALDAE is provided in Algorithm 1 in the appendix, where we also demonstrate that our algorithm has the potential to find the optimal solution if given sufficient time and is able to produce solutions at any point during its execution. Furthermore, the algorithm has desirable properties such as no redundancy in storage and computation, which contribute to its accuracy and speed.

\section{Empirical Evaluation}

We experimentally compare the SALDAE algorithm against representative state-of-the-art CSG algorithms for small and large-scale problems. We compare the solution quality (for small-scale problems) and the gain rate~\cite{9643288} (for large-scale problems). We ran the PICS~\cite{10098066} algorithm with a number of processes set to 20 and we ran CSG-UCT~\cite{wu2020monte} with a number of iterations set to $10^{2}$, as suggested by the authors of the papers. We implemented our algorithm in Java and in the comparisons we used the codes provided by the authors of PICS and CSG-UCT, which are also written in Java. The algorithms were run on an Intel Xeon 2.30GHz E5-2650 CPU with 256GB of RAM. 

To generate the problem instances, we considered the following value distributions: Agent-based Uniform~\cite{rahwan2012hybrid}, Agent-based Normal, Beta, Exponential, Gamma~\cite{michalak2016hybrid}, Normal~\cite{rahwan2007near}, Uniform~\cite{larson2000anytime}, Pascal and Zipf~\cite{changder2020odss}. The result for each value distribution was produced by computing the average result from 50 generated problem instances per value distribution. The best strategies selection for SALDAE and the hyperparameters are explained in the appendix. 

\subsection{Small-Scale Benchmarks}

In this subsection, we investigate how our algorithm compares to the state-of-the-art algorithms in solving small-scale problems with small numbers of agents. We run the algorithms on the nine value distributions and computed the solution quality achieved by the algorithms. Note that the algorithms behave differently depending on the value distributions. The solutions obtained by the algorithms are compared to the solutions provided by ODP-IP~\cite{michalak2016hybrid}, which always yields the optimal solutions. Figure \ref{solutionQuality} shows the results. In these experiments, we set the number of search agents of SALDAE to 10 and we stopped the algorithms when they finish or at the time when ODP-IP finds the optimal solution in case they take more time to finish than ODP-IP. Moreover, the number of cores used for each algorithm was matched to the number of processes it utilized. Specifically, SALDAE was run on 10 cores, PICS on 20 cores, and CSG-UCT on 1 core. 
To ensure fairness in the comparison with CSG-UCT, we added results in the appendix of SALDAE using only one search agent and one core, demonstrating that even with limited resources, SALDAE still outperforms the other algorithm. 

Figure \ref{solutionQuality} clearly shows that our algorithm provides higher quality solutions than the other algorithms, which demonstrates its effectiveness. For example, with the Exponential distribution, SALDAE produces up to $28\%$ higher solution quality than CSG-UCT and up to $15\%$ higher solution quality than PICS. A notable exception is the Agent-based Normal, for which there is almost a tie (see the appendix). 

In all value distributions, our algorithm provides optimal solutions more frequently than the other algorithms. For example with the Exponential distribution, SALDAE provides optimal solutions in 91\% of the cases. With the same instances of Exponential, PICS and CSG-UCT produce the optimal solution in 7\% and 5\% of the Exponential instances, respectively 
(see Figure \ref{optimalrate}). A more detailed experimental results are shown in the appendix. In summary, our algorithm consistently produces higher quality solutions than the PICS and CSG-UCT algorithms, while at the same time providing optimal solutions more frequently than the other algorithms.

\begin{figure}[t]
\begin{center}
\small
\resizebox{0.995\columnwidth}{4.9cm}{%
         \centering
         \begin{tikzpicture}
	\begin{axis}[	grid= major ,
            title=(g) Success rate of the algorithms ,
			width=1.0\textwidth ,
            xlabel = {Distribution} ,
			ylabel = {Success rate ($\%$)} ,
            width=9cm,height=5.3cm,
			ytick={0,20,40,70,80,90,100},
            yticklabels={0,20,40,70,80,90,100},
            xtick={1,2,3,4,5,6,7,8},
            xticklabels={Normal,Uniform,Beta,Gamma,A-b N,A-b U,Pascal,Zipf},
            xticklabel style={rotate=15},
            label style={font=\large},
			title style={font=\large},
			legend entries={SALDAE,PICS,CSG-UCT},
			legend style={at={(0,0.65)},anchor=north west,opacity=0.6,text opacity = 1}]
            \addplot+[only marks,color = blue,mark=square*,mark options={fill=cyan}] coordinates {(1,26) (2,79) (3,74.3) (4,81.4) (5,29.1) (6,25.6) (7,100) (8,64.3) 
            };
            \addplot+[only marks,color = vert,mark=*,mark options={fill=vert}] coordinates {(1,4.9) (2,7.6) (3,25.1) (4,9.6) (5,6.6) (6,4.4) 
            (7,87) (8,13) 
            };
            \addplot+[only marks,color = red,mark=triangle*,mark options={fill=red}] coordinates {(1,3.8) (2,5.4) (3,25.6) (4,4) (5,7.9) (6,4.8) 
            (7,1) (8,4.9) 
            };
	\end{axis}
    \end{tikzpicture}
         \label{tabresultPICS}
}

             \caption{Success rate of the SALDAE, PICS and CSG-UCT algorithms on 2000 executions per distribution.}
        \label{optimalrate}
\normalsize
\end{center}
\end{figure}
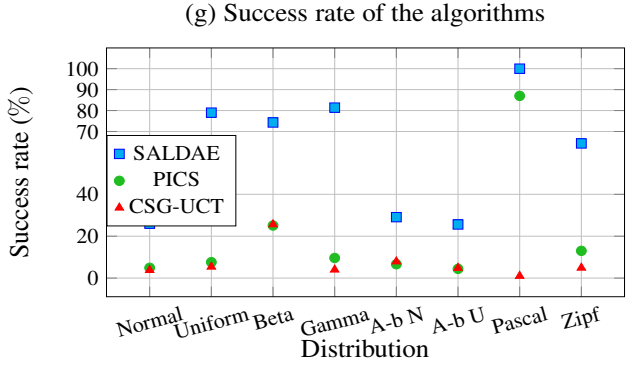

\subsection{Large-Scale Benchmarks}

The fastest optimal algorithm, ODP-IP, only scales to 30 - 40 agents. 
Our algorithm is able to handle large problems with hundreds and thousands of agents. The algorithms PICS and CSG-UCT could also handle large-scale problems. However, for these settings, it is infeasible to guarantee to find an optimal solution due to the exponentiality of the solution space. Hence, we cannot compute the solution quality by comparing the solutions obtained to the optimal solutions. This is why we compare the algorithms using the gain rate, which measures the improvement of the solution achieved by the algorithms relative to the value of the singleton coalition structure. 
The \textit{gain rate} is computed as {\small$\frac{\frac{v(\mathcal{CS})}{v(\mathcal{CS}_{s})}}{\max_{i}(\frac{v(\mathcal{CS}^{+}_{i})}{v(\mathcal{CS}_{s})})}$}, where $i \in \{\text{SALDAE, PICS, CSG-UCT}\}$ and  $v(\mathcal{CS}_{s})$ is the value of the singleton coalition structure, which represents a partition into $n$ coalitions, each containing a single agent, and $v(\mathcal{CS}^{+}_{i})$ is the value of the best solutions provided by the algorithm $i$. We also considered a Disaster Response distribution introduced in~\cite{wu2020monte}, 
in which hundreds of human responders must be quickly organized into teams to coordinate their evacuation and rescue actions. 

Figure \ref{largeScaleResults} shows the results of our large-scale benchmarks. 
We compared SALDAE to PICS, FACS \cite{9643288}, and CSG-UCT. The result of each experiment was produced by evaluating all the algorithms on instances of the different value distributions. To make sure that the algorithms competed on a similar search time, we used the same time limit for all the algorithms. First, we can see a clear general trend that SALDAE and PICS outperform the other algorithms. A notable exception is the Pascal distribution, for which there is almost a tie for less than 100 agents. We can also see that the SALDAE algorithm outperforms the PICS algorithm in a majority of the tests, most notably on problem instances of the Gamma, Exponential, Normal, Disaster Response, and Electric Vehicles Allocation distributions. However, PICS performed on par in Beta and Uniform distributions and 
in some tests of Agent-based Normal distributions (see the appendix). SALDAE’s superior performance comes from the different heuristics that guide the search. For example, the bridging path strategies allow it to search other areas of the solution space and find better solutions quicker. Pathfinding between best solutions aims to explore promising regions of the search space. It provides a form of strategic neighborhood search and diversification to the greedy algorithm’s approach. This helps SALDAE explore better solutions faster than alternative methods.

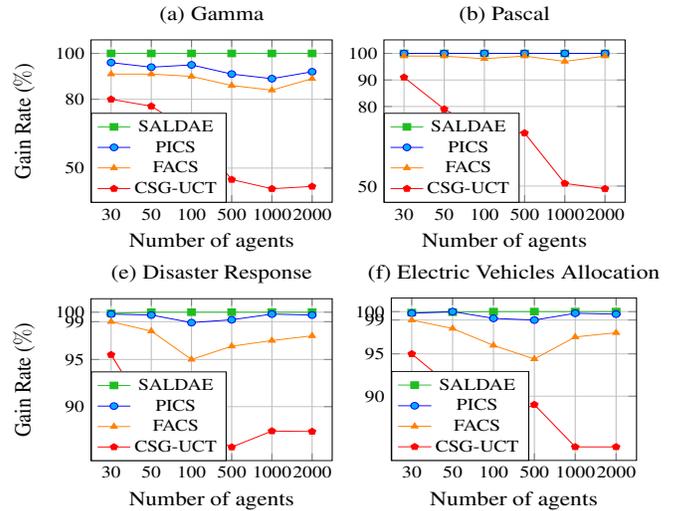
\begin{figure}[t]
\begin{center}
\normalsize
\resizebox{0.46\columnwidth}{3.4cm}{%
     \begin{subfigure}{0.295\textwidth}
         \centering
         \begin{tikzpicture}
	\begin{axis}[	grid= major ,
            title=(a) Gamma ,
			width=0.5\textwidth ,
            xlabel = {Number of agents} ,
			ylabel = {Gain Rate ($\%$)} ,
            width=6cm,height=5cm,
			ytick={50,80,100},
            yticklabels={50,80,100},
            xtick={1,2,3,4,5,6},
            xticklabels={30,50,100,500,1000,2000},
            label style={font=\Large},
			title style={font=\Large},
            legend entries={SALDAE, PICS, FACS, CSG-UCT},
            legend style={at={(0,0.55)},anchor=north west,opacity=0.6,text opacity = 1}]
            \addplot[color = vert,mark=square*,mark options={fill=vert}] coordinates {(1,100) (2,100) (3,100) (4,100) (5,100) (6,100) };
            \addplot[color = blue,mark=*,mark options={fill=cyan}] coordinates {(1,96) (2,94) (3,95) (4,91) (5,89) (6,92) };
            \addplot[color = orange,mark=triangle*,mark options={fill=orange}] coordinates {(1,91) (2,91) (3,90) (4,86) (5,84) (6,89) };
            \addplot[mark=pentagon*,mark options={fill=red},color=red] coordinates {(1,80) (2,77) (3,66) (4,45) (5,41) (6,42) };
	\end{axis}
    \end{tikzpicture}\\
         \label{fig:three sin x}
     \end{subfigure}
}
\hfill
\:\:\:
\hfill
\resizebox{0.46\columnwidth}{3.4cm}{%
     \begin{subfigure}{0.295\textwidth}
         \centering
         \begin{tikzpicture}
	\begin{axis}[	grid= major ,
            title=(b) Pascal ,
			width=0.8\textwidth ,
            xlabel = {Number of agents} ,
            width=6cm,height=5cm,
			ytick={50,80,90,100},
            yticklabels={50,80,90,100},
            xtick={1,2,3,4,5,6},
            xticklabels={30,50,100,500,1000,2000},
            label style={font=\Large},
			title style={font=\Large},
            legend entries={SALDAE, PICS, FACS, CSG-UCT},
            legend style={at={(0,0.55)},anchor=north west,opacity=0.6,text opacity = 1}]
            \addplot[color = vert,mark=square*,mark options={fill=vert}] coordinates {(1,100) (2,100) (3,100) (4,100) (5,100) (6,100) };
            \addplot[color = blue,mark=*,mark options={fill=cyan}] coordinates {(1,100) (2,100) (3,100) (4,100) (5,100) (6,100) 
            };
            \addplot[color = orange,mark=triangle*,mark options={fill=orange}] coordinates {(1,99) (2,99) (3,98) (4,99) (5,97) (6,99) 
            };
            \addplot[mark=pentagon*,mark options={fill=red},color=red] coordinates {(1,91) (2,79) (3,71) (4,70) (5,51) (6,49)
            };
	\end{axis}
    \end{tikzpicture}\\
         \label{fig:three sin x}
     \end{subfigure}
}
     \:
     \resizebox{0.46\columnwidth}{3.4cm}{%
     \begin{subfigure}{0.295\textwidth}
         \centering
         \begin{tikzpicture}
	\begin{axis}[	grid= major ,
            title=(e) Disaster Response ,
			width=0.8\textwidth ,
            xlabel = {Number of agents} ,
			ylabel = {Gain Rate ($\%$)} ,
            width=6cm,height=5cm,
			ytick={90,95,99,100},
            yticklabels={90,95,99,100},
            xtick={1,2,3,4,5,6},
            xticklabels={30,50,100,500,1000,2000},
            label style={font=\Large},
			title style={font=\Large},
            legend entries={SALDAE, PICS, FACS, CSG-UCT},
            legend style={at={(0,0.55)},anchor=north west,opacity=0.4,text opacity = 1}]
            \addplot[color = vert,mark=square*,mark options={fill=vert}] coordinates {(1,99.89) (2,100) (3,100) (4,100) (5,100) (6,100) };
            \addplot[color = blue,mark=*,mark options={fill=cyan}] coordinates {(1,99.8) (2,99.7) (3,98.9) (4,99.2) (5,99.8) (6,99.7) 
            };
            \addplot[color = orange,mark=triangle*,mark options={fill=orange}] coordinates {(1,99) (2,98) (3,95) (4,96.4) (5,97) (6,97.5) 
            };
            \addplot[mark=pentagon*,mark options={fill=red},color=red] coordinates {(1,95.50) (2,88.65) (3,87.42) (4,85.72) (5,87.42) (6,87.37) };
	\end{axis}
    \end{tikzpicture}\\
         \label{fig:three sin x}
     \end{subfigure}
}
\hfill
\:
\resizebox{0.46\columnwidth}{3.4cm}{%
     \begin{subfigure}{0.295\textwidth}
         \centering
         \begin{tikzpicture}
	\begin{axis}[	grid= major ,
            title=(f) Electric Vehicles Allocation ,
			width=0.8\textwidth ,
            xlabel = {Number of agents} ,
            width=6cm,height=5cm,
			ytick={90,95,99,100},
            yticklabels={90,95,99,100},
            xtick={1,2,3,4,5,6},
            xticklabels={30,50,100,500,1000,2000},
            label style={font=\Large},
			title style={font=\Large},
            legend entries={SALDAE, PICS, FACS, CSG-UCT},
            legend style={at={(0,0.55)},anchor=north west,opacity=0.4,text opacity = 1}]
            \addplot[color = vert,mark=square*,mark options={fill=vert}] coordinates {(1,99.89) (2,99.93) (3,100) (4,100) (5,100) (6,100) };
            \addplot[color = blue,mark=*,mark options={fill=cyan}] coordinates {(1,99.8) (2,100) (3,99.2) (4,99.0) (5,99.8) (6,99.7) 
            };
            \addplot[color = orange,mark=triangle*,mark options={fill=orange}] coordinates {(1,99) (2,98) (3,96) (4,94.4) (5,97) (6,97.5) 
            };
            \addplot[mark=pentagon*,mark options={fill=red},color=red] coordinates {(1,95) (2,91) (3,87) (4,89) (5,84) (6,84) };
	\end{axis}
    \end{tikzpicture}\\
         \label{fig:three sin x}
     \end{subfigure}
     }
             \caption{Gain rate of SALDAE, PICS, FACS and CSG-UCT when run with large numbers of agents.
             }
        \label{largeScaleResults}
\normalsize
\end{center}
\end{figure}


\section{Conclusion}

We presented a multiagent path search inspired algorithm for coalition structure generation. In more detail, we developed an algorithm that utilizes multiple search agents to incrementally explore a search graph using various heuristics to guide the search process. 
Furthermore, we introduced different strategies for connecting the best solutions in order to improve upon them, as well as strategies for resolving conflicts between the search agents. 
The resulting algorithm is anytime and can handle large-scale problems. 
We ran experiments on a variety of different value distributions and our results demonstrate that our algorithm can perform better than existing well-established state-of-the-art algorithms in solving the coalition structure generation problem, often achieving significantly higher solution quality and gain rate.

Future work could explore the adaptation and application of other MAPF algorithms to address
this problem. Another potential avenue for applying MAPF in CSG involves using a different
approach where search agents not only have access to start nodes, as in the current version, but
also to goal nodes. Currently, we lack a clear definition of goal nodes beyond the one representing
the optimal solution, and we are unaware of its position in the graph. However, we have the
flexibility to define additional goal nodes, either randomly or through a predefined function, and
then the algorithm can be used to find paths between these start and goal nodes.

\section*{Acknowledgments}
Tuomas Sandholm's research is supported by the Vannevar Bush Faculty Fellowship ONR N00014-23-1-2876, National Science Foundation grant RI-2312342, ARO award W911NF2210266, and NIH award A240108S001.

\bibliography{aaai25}

\newpage

\onecolumn

\appendix

%



\section{Analysis of SALDAE}

\begin{thm}
Our algorithm is anytime.
\end{thm}

\begin{proof}
To compute the coalition structures, the algorithm starts with an initial coalition structure (a node) and then improves it over time. Notice that, the algorithm evaluates the solutions by comparing their values to that of the currently best solution. Hence, it always maintains the value of the currently best solution. When execution is stopped, it can still return the currently best solution that is better than the previous best ones. Therefore, the algorithm is anytime.
\end{proof}

\begin{property}
    The SALDAE algorithm has the following properties.
    \begin{itemize}
        \item \textbf{P1} Limited redundancy in storage: Each node is stored by at most one search agent.
        \item \textbf{P2} No redundancy in computation: Each node is expanded by one search agent.
    \end{itemize}
\end{property}

P1 is desirable for accuracy. Less redundancy in storage enables the algorithm to store more unique nodes. P2 is desirable for speed.

\textbf{Optimal solution generation discussion. } 
The proposed algorithm has the potential to return the optimal solution if given a sufficient amount of run time. 
Our algorithm operates by iteratively expanding a search graph of coalition structures. At each iteration, the algorithm generates new child nodes by considering different ways of splitting or merging coalitions from the coalition structures that have been visited so far. Our algorithm adds after each iteration new child nodes to the search graph. If the algorithm is given an unlimited amount of time with an unlimited amount of memory space, it can eventually 
add the optimal coalition structure to the graph and produce the optimal solution. 
This is because every possible coalition structure can be reached by a sequence of splits or merges from an existing coalition structure. 
However, in practice, it may not be feasible to run the algorithm for an unlimited amount of time, especially for large-scale problems. Nonetheless, for large-scale problems, it is intractable to guarantee to find the optimal solution in a limited run time.


\begin{proof}[Proof of Observation 1]
A merger is an operation that combines two coalitions into a single coalition. A split is an operation that takes a single coalition and divides it into two coalitions. The level of a node indicates the number of coalitions in the partition. For example, a node at level 1 represents a partition into a single coalition containing all the agents, while a node at level $n$ represents a partition into $n$ coalitions, each containing a single agent.

Now, we can prove the two parts of the lemma separately.

For the first part, we want to show that the bottom node, which represents a partition into a single coalition containing all the agents, can be reached from $\mathcal{N}$ with $l-1$ mergers and that $\mathcal{N}$ can be reached from the bottom node with $l-1$ splits.

To reach the bottom node from $\mathcal{N}$, we can perform $l-1$ mergers, starting with the two coalitions containing the smallest number of agents and continuing until all the agents are in a single coalition. This requires $l-1$ mergers, as there are $l-1$ pairs of coalitions to merge.

To reach $\mathcal{N}$ from the bottom node, we can perform $l-1$ splits, starting with the single coalition containing all the agents and dividing it into two coalitions. We can then continue to divide the remaining coalition into two coalitions until we have $l$ coalitions. This requires $l-1$ splits, as there are $l-1$ coalitions to split.

For the second part of the lemma, we want to show that the top node, which represents a partition into $n$ coalitions, each containing a single agent, can be reached from $\mathcal{N}$ with $n-l$ splits and that $\mathcal{N}$ can be reached from the top node with $n-l$ mergers.

To reach the top node from $\mathcal{N}$, we can perform $n-l$ splits on coalitions that contain at least 2 agents until all the agents are in a single coalition. This requires $n-l$ splits, as there are $n-l$ coalitions to split.

To reach $\mathcal{N}$ from the top node, we can perform $n-l$ splits, starting with two coalitions containing a single agent and continuing until all agents are in the correct coalition. This requires $n-l$ splits, as there are $n-l$ pairs of coalitions to merge.

Therefore, the lemma is proved.

\end{proof}

\begin{proof}[Proof of Observation 2]
Let $l_{1}$ and $l_{2}$ be the levels of $\mathcal{CS}_{1}$ and $\mathcal{CS}_{2}$, respectively. We will show that for every $l_{1},l_{2} \in \{1,...,n\}$ there exists a path of size at most $n-1$. We split the proof into two cases. In all cases, starting from $\mathcal{CS}_{1}$, a path needs to visit either the top node through consecutive splits or the bottom node through consecutive merges, before heading to $\mathcal{CS}_{2}$. We will indicate which one to choose for each case. 
The shortest path between $\mathcal{CS}_{1}$ and the top node is of size $n - l_{1}$ and the shortest path between $\mathcal{CS}_{2}$ and the top node is of size $n - l_{2}$ by Lemma~1. Also, The shortest path between $\mathcal{CS}_{1}$ and the bottom node is of size $l_{1}-1$ and the shortest path between $\mathcal{CS}_{2}$ and the bottom node is of size $l_{2}-1$ by Lemma~1. The shortest path between $\mathcal{CS}_{1}$ and $\mathcal{CS}_{2}$ that visits the top node is of size $n-l_{1} + n-l_{2}$. The shortest path between $\mathcal{CS}_{1}$ and $\mathcal{CS}_{2}$ that visits the bottom node is of size $l_{1}-1 + l_{2}-1 = l_{1} + l_{2}-2$. 

\begin{itemize}
    \item \textbf{Case 1:} $n-l_{1} + n-l_{2} > n-1$. 
    In this case, the size of the path between $\mathcal{CS}_{1}$ and $\mathcal{CS}_{2}$ that visits the top node is greater than $n-1$. We will now prove that in this case, the shortest path between $\mathcal{CS}_{1}$ and $\mathcal{CS}_{2}$ that visits the bottom node is of size at most $n-1$. We supposed that $n-l_{1} + n-l_{2} > n-1$. Hence, $n+n>n-1 + l_{1} + l_{2}$. $l_{1} + l_{2} < n+1$, and thus, $l_{1} + l_{2} - 2 < n-1$. This means that when the size of the path that visits the top node is greater than $n-1$, the size of the path that visits the bottom node is lower than $n-1$, so Lemma~2 holds.
    \item \textbf{Case 2:} $l_{1} + l_{2} - 2 > n-1$ (i.e. the size of the path that visits the bottom node is greater than $n-1$). By the same logic as the previous case, we will now prove that in this case, the shortest path between $\mathcal{CS}_{1}$ and $\mathcal{CS}_{2}$ that visits the bottom node is of size at most $n-1$. We supposed that $l_{1} + l_{2} - 2 > n-1$. Hence, $n+l_{1} + l_{2} - 2 > n+n-1$. $n+n-1-l_{1}-l_{2} < n-2$. Thus, $n-l_{1}+n-l_{2}<n-1$. This means that when the size of the path that visits the bottom node is greater than $n-1$, the size of the path that visits the top node is lower than $n-1$, so Lemma~2 holds. 
\end{itemize}

As a result, there is always a path of size at most $n-1$ between two coalition structures.





\end{proof}


Below is a theorem that holds for the SPLIT-THEN-MERGE and MERGE-THEN-SPLIT strategies. The proof is omitted as it follows the proof of Lemma~2.

\begin{thm}
The SPLIT-THEN-MERGE and MERGE-THEN-SPLIT strategies guarantee to return a shortest path that visits either the top or the bottom node and the size of this path is at most $n-1$.
\end{thm}

The APPROACH-THEN-SWAP strategy also guarantees that a path is found. However, the size of this path can be large. 

\section{Additional Figures}

Figures \ref{CSG1} and \ref{CSG2} show illustration examples of steps 1 and 2 of SALDAE. 

Figures \ref{splitthenmerge}, \ref{mergethensplit}, and \ref{approachthenswap} show illustration examples of the strategies of connecting the best solutions used by SALDAE.

\begin{figure}[h!]
\centering
\normalsize
\resizebox{0.95\columnwidth}{8.1cm}{%
\begin{tikzpicture}
\large
\node[draw,rectangle,rounded corners=3pt,black!30] (a)at(7.5,0){$\{1\},\{2\},\{3\},\{4\}$};
\node[draw,rectangle,rounded corners=3pt,black!30] (b)at(0,-1.5){$\{1\},\{2\},\{3,4\}$};
\node[draw,rectangle,rounded corners=3pt,black!30] (c)at(3,-1.5){$\{3\},\{4\},\{1,2\}$};
\node[draw,rectangle,rounded corners=3pt,black!30] (d)at(6,-1.5){$\{1\},\{3\},\{2,4\}$};
\node[draw,rectangle,rounded corners=3pt,black!30] (e)at(9,-1.5){$\{2\},\{4\},\{1,3\}$};
\node[draw,rectangle,rounded corners=3pt,black!30] (f)at(12,-1.5){$\{1\},\{4\},\{2,3\}$};
\node[draw,rectangle,rounded corners=3pt,black!30] (g)at(15,-1.5){$\{2\},\{3\},\{1,4\}$};

\node[draw,rectangle,rounded corners=3pt,vert] (h)at(-1.5,-3){\color{black}$\{1\},\{2,3,4\}$};
\node[draw,rectangle,rounded corners=3pt,vert] (i)at(1.5,-3){\color{black}$\{1,2\},\{3,4\}$};
\node[draw,rectangle,rounded corners=3pt,vert] (j)at(4.5,-3){\color{black}$\{2\},\{1,3,4\}$};
\node[draw,rectangle,rounded corners=3pt,vert] (k)at(7.5,-3){\color{black}$\{1,3\},\{2,4\}$};
\node[draw,rectangle,rounded corners=3pt,vert] (l)at(10.5,-3){\color{black}$\{3\},\{1,2,4\}$};
\node[draw,rectangle,rounded corners=3pt,vert] (m)at(13.5,-3){\color{black}$\{1,4\},\{2,3\}$};
\node[draw,rectangle,rounded corners=3pt,vert] (n)at(16.5,-3){\color{black}$\{4\},\{1,2,3\}$};
\node[draw,rectangle,rounded corners=3pt,red] (o)at(7.5,-4.5){\color{black}$\{1,2,3,4\}$};

\draw[-,>=latex,blue!30,dashed] (0,-1.17) -- (7.5,-0.32);
\draw[-,>=latex,blue!30,dashed] (3,-1.17) -- (7.5,-0.32);
\draw[-,>=latex,blue!30,dashed] (6,-1.17) -- (7.5,-0.32);
\draw[-,>=latex,blue!30,dashed] (9,-1.17) -- (7.5,-0.32);
\draw[->,>=latex,bleu!30,dashed] (12,-1.17) -- (7.5,-0.32);
\draw[-,>=latex,blue!30,dashed] (15,-1.17) -- (7.5,-0.32);

\draw[-,>=latex,black!30,dashed] (-1.5,-2.65) -- (0,-1.83);
\draw[-,>=latex,black!30,dashed] (-1.5,-2.65) -- (6,-1.83);
\draw[-,>=latex,black!30,dashed] (-1.5,-2.65) -- (12,-1.83);
\draw[-,>=latex,teal!30,thick,dashed] (1.5,-2.65) -- (0,-1.83);
\draw[-,>=latex,teal!30,thick,dashed] (1.5,-2.65) -- (3,-1.83);
\draw[-,>=latex,black!30,dashed] (4.5,-2.65) -- (0,-1.83);
\draw[-,>=latex,black!30,dashed] (4.5,-2.65) -- (9,-1.83);
\draw[-,>=latex,black!30,dashed] (4.5,-2.65) -- (15,-1.83);
\draw[-,>=latex,teal!30,thick,dashed] (7.5,-2.65) -- (6,-1.83);
\draw[-,>=latex,teal!30,thick,dashed] (7.5,-2.65) -- (9,-1.83);
\draw[-,>=latex,black!30,dashed] (10.5,-2.65) -- (3,-1.83);
\draw[-,>=latex,black!30,dashed] (10.5,-2.65) -- (6,-1.83);
\draw[-,>=latex,black!30,dashed] (10.5,-2.65) -- (15,-1.83);
\draw[-,>=latex,teal!30,thick,dashed] (13.5,-2.65) -- (12,-1.83);
\draw[-,>=latex,teal!30,thick,dashed] (13.5,-2.65) -- (15,-1.83);
\draw[-,>=latex,black!30,dashed] (16.5,-2.65) -- (3,-1.83);
\draw[-,>=latex,black!30,dashed] (16.5,-2.65) -- (9,-1.83);
\draw[-,>=latex,black!30,dashed] (16.5,-2.65) -- (12,-1.83);

\draw[->,>=latex,vert,dashed] (7.5,-4.17) -- (-1.5,-3.35);
\draw[->,>=latex,vert,dashed] (7.5,-4.17) -- (1.5,-3.35);
\draw[->,>=latex,vert,dashed] (7.5,-4.17) -- (4.5,-3.35);
\draw[->,>=latex,vert,dashed] (7.5,-4.17) -- (7.5,-3.35);
\draw[->,>=latex,vert,dashed] (7.5,-4.17) -- (10.5,-3.35);
\draw[->,>=latex,vert,dashed] (7.5,-4.17) -- (13.5,-3.35);
\draw[->,>=latex,vert,dashed] (7.5,-4.17) -- (16.5,-3.35);
            \end{tikzpicture}%
        }
\caption{An illustration of the generation step with 4 agents. The start node is the red node. The child nodes are represented by the green nodes. All of these child nodes are candidates to become the new start node. At this stage, only a partial graph is built.} \label{CSG1}
\end{figure}
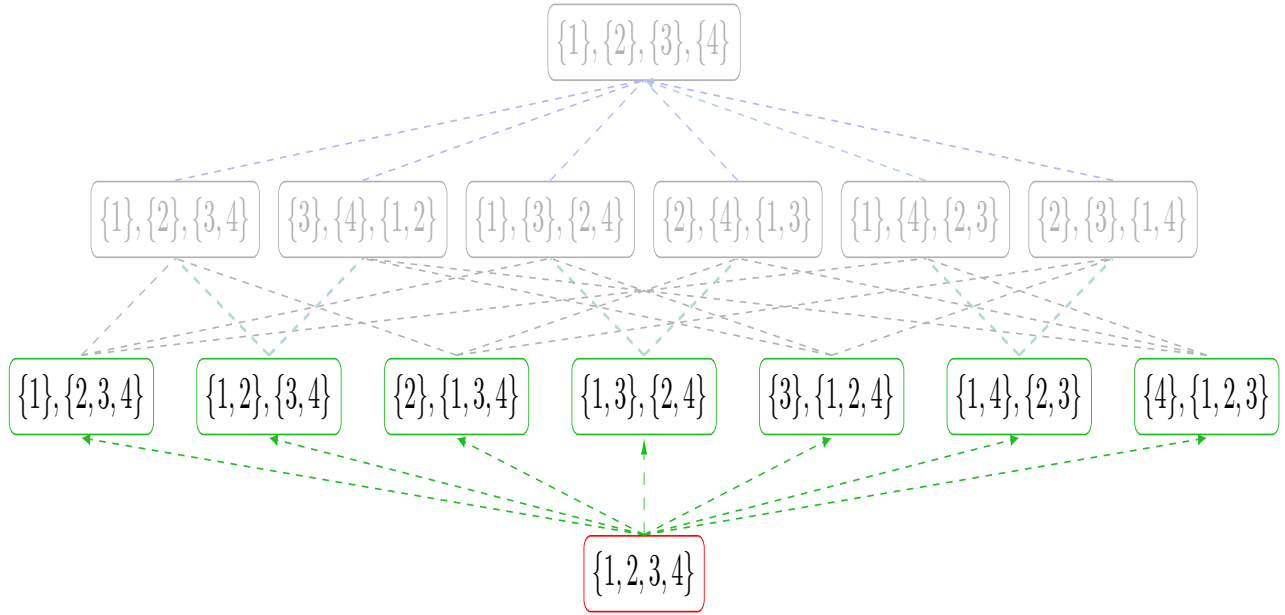

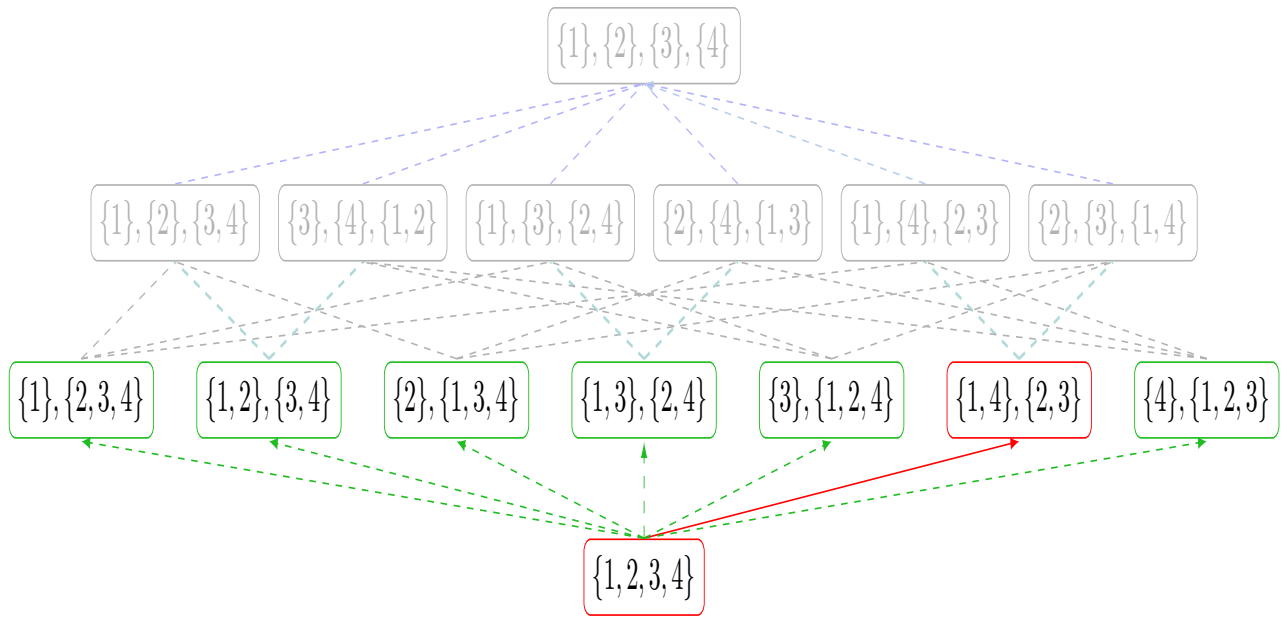
\begin{figure}[h!]
\centering
\normalsize
\resizebox{0.95\columnwidth}{8.1cm}{%
\begin{tikzpicture}
\large
\node[draw,rectangle,rounded corners=3pt,black!30] (a)at(7.5,0){$\{1\},\{2\},\{3\},\{4\}$};
\node[draw,rectangle,rounded corners=3pt,black!30] (b)at(0,-1.5){$\{1\},\{2\},\{3,4\}$};
\node[draw,rectangle,rounded corners=3pt,black!30] (c)at(3,-1.5){$\{3\},\{4\},\{1,2\}$};
\node[draw,rectangle,rounded corners=3pt,black!30] (d)at(6,-1.5){$\{1\},\{3\},\{2,4\}$};
\node[draw,rectangle,rounded corners=3pt,black!30] (e)at(9,-1.5){$\{2\},\{4\},\{1,3\}$};
\node[draw,rectangle,rounded corners=3pt,black!30] (f)at(12,-1.5){$\{1\},\{4\},\{2,3\}$};
\node[draw,rectangle,rounded corners=3pt,black!30] (g)at(15,-1.5){$\{2\},\{3\},\{1,4\}$};

\node[draw,rectangle,rounded corners=3pt,vert] (h)at(-1.5,-3){\color{black}$\{1\},\{2,3,4\}$};
\node[draw,rectangle,rounded corners=3pt,vert] (i)at(1.5,-3){\color{black}$\{1,2\},\{3,4\}$};
\node[draw,rectangle,rounded corners=3pt,vert] (j)at(4.5,-3){\color{black}$\{2\},\{1,3,4\}$};
\node[draw,rectangle,rounded corners=3pt,vert] (k)at(7.5,-3){\color{black}$\{1,3\},\{2,4\}$};
\node[draw,rectangle,rounded corners=3pt,vert] (l)at(10.5,-3){\color{black}$\{3\},\{1,2,4\}$};
\node[draw,rectangle,rounded corners=3pt,red] (m)at(13.5,-3){\color{black}$\{1,4\},\{2,3\}$};
\node[draw,rectangle,rounded corners=3pt,vert] (n)at(16.5,-3){\color{black}$\{4\},\{1,2,3\}$};
\node[draw,rectangle,rounded corners=3pt,red] (o)at(7.5,-4.5){\color{black}$\{1,2,3,4\}$};

\draw[-,>=latex,blue!30,dashed] (0,-1.17) -- (7.5,-0.32);
\draw[-,>=latex,blue!30,dashed] (3,-1.17) -- (7.5,-0.32);
\draw[-,>=latex,blue!30,dashed] (6,-1.17) -- (7.5,-0.32);
\draw[-,>=latex,blue!30,dashed] (9,-1.17) -- (7.5,-0.32);
\draw[->,>=latex,bleu!30,dashed] (12,-1.17) -- (7.5,-0.32);
\draw[-,>=latex,blue!30,dashed] (15,-1.17) -- (7.5,-0.32);

\draw[-,>=latex,black!30,dashed] (-1.5,-2.65) -- (0,-1.83);
\draw[-,>=latex,black!30,dashed] (-1.5,-2.65) -- (6,-1.83);
\draw[-,>=latex,black!30,dashed] (-1.5,-2.65) -- (12,-1.83);
\draw[-,>=latex,teal!30,thick,dashed] (1.5,-2.65) -- (0,-1.83);
\draw[-,>=latex,teal!30,thick,dashed] (1.5,-2.65) -- (3,-1.83);
\draw[-,>=latex,black!30,dashed] (4.5,-2.65) -- (0,-1.83);
\draw[-,>=latex,black!30,dashed] (4.5,-2.65) -- (9,-1.83);
\draw[-,>=latex,black!30,dashed] (4.5,-2.65) -- (15,-1.83);
\draw[-,>=latex,teal!30,thick,dashed] (7.5,-2.65) -- (6,-1.83);
\draw[-,>=latex,teal!30,thick,dashed] (7.5,-2.65) -- (9,-1.83);
\draw[-,>=latex,black!30,dashed] (10.5,-2.65) -- (3,-1.83);
\draw[-,>=latex,black!30,dashed] (10.5,-2.65) -- (6,-1.83);
\draw[-,>=latex,black!30,dashed] (10.5,-2.65) -- (15,-1.83);
\draw[-,>=latex,teal!30,thick,dashed] (13.5,-2.65) -- (12,-1.83);
\draw[-,>=latex,teal!30,thick,dashed] (13.5,-2.65) -- (15,-1.83);
\draw[-,>=latex,black!30,dashed] (16.5,-2.65) -- (3,-1.83);
\draw[-,>=latex,black!30,dashed] (16.5,-2.65) -- (9,-1.83);
\draw[-,>=latex,black!30,dashed] (16.5,-2.65) -- (12,-1.83);

\draw[->,>=latex,vert,dashed] (7.5,-4.17) -- (-1.5,-3.35);
\draw[->,>=latex,vert,dashed] (7.5,-4.17) -- (1.5,-3.35);
\draw[->,>=latex,vert,dashed] (7.5,-4.17) -- (4.5,-3.35);
\draw[->,>=latex,vert,dashed] (7.5,-4.17) -- (7.5,-3.35);
\draw[->,>=latex,vert,dashed] (7.5,-4.17) -- (10.5,-3.35);
\draw[->,>=latex,red] (7.5,-4.17) -- (13.5,-3.35);
\draw[->,>=latex,vert,dashed] (7.5,-4.17) -- (16.5,-3.35);
            \end{tikzpicture}%
        }
\caption{An illustration of the selection procedure. All the green nodes are candidate child nodes. However, the coalition structure $\{\{a_{1},a_{4}\},\{a_{2},a_{3}\}\}$ has, we assume, the highest value and hence it is selected.} \label{CSG2}
\end{figure}

\begin{figure}[h!]
\centering
\normalsize
\resizebox{0.99\columnwidth}{8.9cm}{%
\begin{tikzpicture}
\large
\node[draw,rectangle,rounded corners=3pt,blue] (a0)at(7.5,0){\color{black}$\{1\},\{2\},\{3\},\{4\},\{5\},\{6\}$};

\large
\node[] (k)at(-4.15,-1.5){\huge \color{black}...};
\node[draw,rectangle,rounded corners=3pt,blue] (b1)at(-1.4,-1.5){\color{black}$\{1\},\{2\},\{4\},\{6\},\{3,5\}$};
\node[draw,rectangle,rounded corners=3pt,black] (c1)at(3.05,-1.5){\color{black}$\{3\},\{4\},\{5\},\{6\}\{1,2\}$};
\node[draw,rectangle,rounded corners=3pt,black] (d1)at(7.5,-1.5){\color{black}$\{1\},\{2\},\{4\},\{6\},\{3,5\}$};
\node[draw,rectangle,rounded corners=3pt,blue] (e1)at(11.95,-1.5){\color{black}$\{2\},\{4\},\{5\},\{6\},\{1,3\}$};
\node[] (k)at(14.45,-1.5){\huge \color{black}...};

\large
\node[] (k)at(-5.0,-3){\huge \color{black}...};
\node[draw,rectangle,rounded corners=3pt,black] (a2)at(-2.5,-3){\color{black}$\{1\},\{2\},\{3,5\},\{4,6\}$};
\node[draw,rectangle,rounded corners=3pt,black] (b2)at(1.5,-3){\color{black}$\{3\},\{5\},\{6\},\{1,2,4\}$};
\node[draw,rectangle,rounded corners=3pt,vert] (c2)at(5.5,-3){\color{black}$\{4\},\{6\}\{1,2\},\{3,5\}$};
\node[draw,rectangle,rounded corners=3pt,black] (d2)at(9.5,-3){\color{black}$\{2\},\{4\},\{6\},\{1,3,5\}$};
\node[draw,rectangle,rounded corners=3pt,blue] (e2)at(13.5,-3){\color{black}$\{4\},\{6\},\{1,3\},\{2,5\}$};
\node[] (k)at(16.0,-3){\huge \color{black}...};

\large
\node[] (k)at(-5.95,-4.5){\huge \color{black}...};
\node[draw,rectangle,rounded corners=3pt,black] (a3)at(-3.45,-4.5){\color{black}$\{1\},\{2\},\{3,4,5,6\}$};
\node[draw,rectangle,rounded corners=3pt,black] (b3)at(0.2,-4.5){\color{black}$\{6\},\{3,5\},\{1,2,4\}$};
\node[draw,rectangle,rounded corners=3pt,black] (c3)at(3.85,-4.5){\color{black}$\{1,2\},\{3,5\},\{4,6\}$};
\node[draw,rectangle,rounded corners=3pt,black] (d3)at(7.5,-4.5){\color{black}$\{4\},\{3,5\},\{1,2,6\}$};
\node[draw,rectangle,rounded corners=3pt,blue] (e3)at(11.15,-4.5){\color{black}$\{6\},\{2,5\},\{1,3,4\}$};
\node[draw,rectangle,rounded corners=3pt,black] (f3)at(14.8,-4.5){\color{black}$\{2\},\{5\},\{1,3,4,6\}$};
\node[] (k)at(17.3,-4.5){\huge \color{black}...};

\Large
\node[] (k)at(-5.5,-6){\huge \color{black}...};
\node[draw,rectangle,rounded corners=3pt,black] (a4)at(-3.0,-6){\color{black}$\{1,2,3\},\{4,5,6\}$};
\node[draw,rectangle,rounded corners=3pt,black] (b4)at(0.5,-6){\color{black}$\{4\},\{1,2,3,5,6\}$};
\node[draw,rectangle,rounded corners=3pt,black] (c4)at(4.0,-6){\color{black}$\{1,2,6\},\{3,4,5\}$};
\node[draw,rectangle,rounded corners=3pt,black] (d4)at(7.5,-6){\color{black}$\{6\},\{1,2,3,4,5\}$};
\node[draw,rectangle,rounded corners=3pt,red] (e4)at(11.0,-6){\color{black}$\{2,5\},\{1,3,4,6\}$};
\node[draw,rectangle,rounded corners=3pt,black] (f4)at(14.5,-6){\color{black}$\{1,3,4\},\{2,5,6\}$};
\node[] (k)at(17.0,-6){\huge \color{black}...};

\Large
\node[draw,rectangle,rounded corners=3pt,black] (a5)at(7.5,-7.5){\color{black}$\{1,2,3,4,5,6\}$};

\draw[->,>=latex,blue!90] (e4) -- (e3);
\draw[->,>=latex,blue!90] (e3) -- (e2);
\draw[->,>=latex,blue!90] (e2) -- (e1);

\draw[-,>=latex,black!30,dashed] (e4) -- (a5);

\draw[-,>=latex,black!30,dashed] (a5) -- (a4.south);
\draw[-,>=latex,black!30,dashed] (a5) -- (b4);
\draw[-,>=latex,black!30,dashed] (a5) -- (d4);
\draw[-,>=latex,black!30,dashed] (a5) -- (f4);

\draw[-,>=latex,black!30,dashed] (b4) -- (d3);
\draw[-,>=latex,black!30,dashed] (d4) -- (b3);
\draw[-,>=latex,black!30,dashed] (d4) -- (e3);
\draw[-,>=latex,black!30,dashed] (e4) -- (f3);
\draw[-,>=latex,black!30,dashed] (f4) -- (e3);

\draw[-,>=latex,black!30,dashed] (a3) -- (a2);
\draw[-,>=latex,black!30,dashed] (b3) -- (b2);
\draw[-,>=latex,black!30,dashed] (b3) -- (c2);
\draw[-,>=latex,black!30,dashed] (c3) -- (c2);
\draw[-,>=latex,black!30,dashed] (e3) -- (d2);

\draw[-,>=latex,black!30,dashed] (a2) -- (d1);
\draw[-,>=latex,black!30,dashed] (b2) -- (c1);
\draw[-,>=latex,black!30,dashed] (c2) -- (c1);
\draw[-,>=latex,black!30,dashed] (d1) -- (c2);
\draw[-,>=latex,black!30,dashed] (d2) -- (d1);
\draw[-,>=latex,black!30,dashed] (e1) -- (d2);

\draw[-,>=latex,black!30,dashed] (c1) -- (a0);
\draw[-,>=latex,black!30,dashed] (d1) -- (a0);

\draw[->,>=latex,blue!90] (e1) -- (a0);
\draw[->,>=latex,blue!90] (a0) -- (b1);
\draw[->,>=latex,blue!90] (b1) -- (c2);

\draw[-,>=latex,black!30,dashed] (e4) -- (a5);
\draw[-,>=latex,black!30,dashed] (a5) -- (c4);
\draw[-,>=latex,black!30,dashed] (c4) -- (d3);
\draw[-,>=latex,black!30,dashed] (d3) -- (c2);


            \end{tikzpicture}%
        }
\caption{An illustration of the SPLIT-THEN-MERGE strategy with a partial graph of 6 agents. The last best solution is the red node and the new best one is the green node. The nodes of the path between the two best solutions that visit the top node are represented by the blue nodes. } \label{splitthenmerge}
\end{figure}
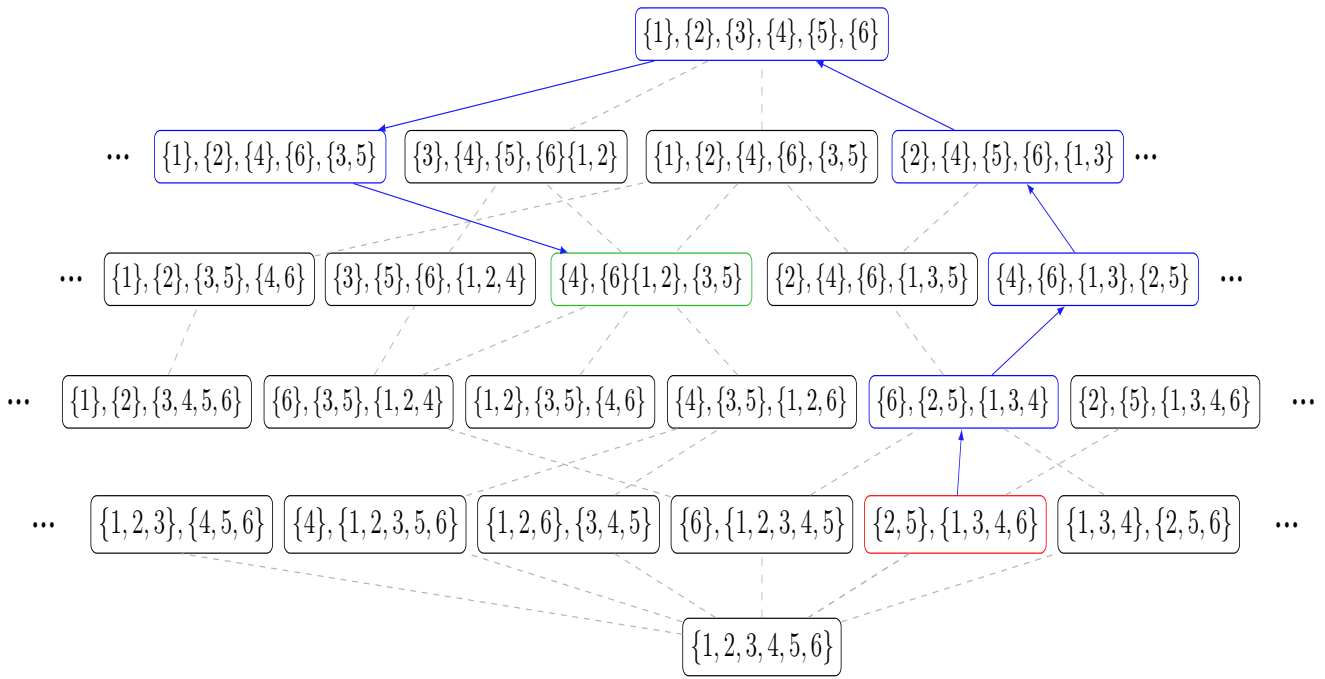

\begin{figure}[h!]
\centering
\normalsize
\resizebox{0.99\columnwidth}{8.9cm}{%
\begin{tikzpicture}
\large
\node[draw,rectangle,rounded corners=3pt,black] (a0)at(7.5,0){\color{black}$\{1\},\{2\},\{3\},\{4\},\{5\},\{6\}$};

\large
\node[] (k)at(-4.15,-1.5){\huge \color{black}...};
\node[draw,rectangle,rounded corners=3pt,black] (b1)at(-1.4,-1.5){\color{black}$\{1\},\{2\},\{4\},\{6\},\{3,5\}$};
\node[draw,rectangle,rounded corners=3pt,black] (c1)at(3.05,-1.5){\color{black}$\{3\},\{4\},\{5\},\{6\}\{1,2\}$};
\node[draw,rectangle,rounded corners=3pt,black] (d1)at(7.5,-1.5){\color{black}$\{1\},\{2\},\{4\},\{6\},\{3,5\}$};
\node[draw,rectangle,rounded corners=3pt,black] (e1)at(11.95,-1.5){\color{black}$\{2\},\{4\},\{5\},\{6\},\{1,3\}$};
\node[] (k)at(14.45,-1.5){\huge \color{black}...};

\large
\node[] (k)at(-5.0,-3){\huge \color{black}...};
\node[draw,rectangle,rounded corners=3pt,black] (a2)at(-2.5,-3){\color{black}$\{1\},\{2\},\{3,5\},\{4,6\}$};
\node[draw,rectangle,rounded corners=3pt,black] (b2)at(1.5,-3){\color{black}$\{3\},\{5\},\{6\},\{1,2,4\}$};
\node[draw,rectangle,rounded corners=3pt,vert] (c2)at(5.5,-3){\color{black}$\{4\},\{6\}\{1,2\},\{3,5\}$};
\node[draw,rectangle,rounded corners=3pt,black] (d2)at(9.5,-3){\color{black}$\{2\},\{4\},\{6\},\{1,3,5\}$};
\node[draw,rectangle,rounded corners=3pt,black] (e2)at(13.5,-3){\color{black}$\{4\},\{6\},\{1,3\},\{2,5\}$};
\node[] (k)at(16.0,-3){\huge \color{black}...};

\large
\node[] (k)at(-5.95,-4.5){\huge \color{black}...};
\node[draw,rectangle,rounded corners=3pt,black] (a3)at(-3.45,-4.5){\color{black}$\{1\},\{2\},\{3,4,5,6\}$};
\node[draw,rectangle,rounded corners=3pt,black] (b3)at(0.2,-4.5){\color{black}$\{6\},\{3,5\},\{1,2,4\}$};
\node[draw,rectangle,rounded corners=3pt,black] (c3)at(3.85,-4.5){\color{black}$\{1,2\},\{3,5\},\{4,6\}$};
\node[draw,rectangle,rounded corners=3pt,blue] (d3)at(7.5,-4.5){\color{black}$\{4\},\{3,5\},\{1,2,6\}$};
\node[draw,rectangle,rounded corners=3pt,black] (e3)at(11.15,-4.5){\color{black}$\{6\},\{2,5\},\{1,3,4\}$};
\node[draw,rectangle,rounded corners=3pt,black] (f3)at(14.8,-4.5){\color{black}$\{2\},\{5\},\{1,3,4,6\}$};
\node[] (k)at(17.3,-4.5){\huge \color{black}...};

\Large
\node[] (k)at(-5.5,-6){\huge \color{black}...};
\node[draw,rectangle,rounded corners=3pt,black] (a4)at(-3.0,-6){\color{black}$\{1,2,3\},\{4,5,6\}$};
\node[draw,rectangle,rounded corners=3pt,black] (b4)at(0.5,-6){\color{black}$\{4\},\{1,2,3,5,6\}$};
\node[draw,rectangle,rounded corners=3pt,blue] (c4)at(4.0,-6){\color{black}$\{1,2,6\},\{3,4,5\}$};
\node[draw,rectangle,rounded corners=3pt,black] (d4)at(7.5,-6){\color{black}$\{6\},\{1,2,3,4,5\}$};
\node[draw,rectangle,rounded corners=3pt,red] (e4)at(11.0,-6){\color{black}$\{2,5\},\{1,3,4,6\}$};
\node[draw,rectangle,rounded corners=3pt,black] (f4)at(14.5,-6){\color{black}$\{1,3,4\},\{2,5,6\}$};
\node[] (k)at(17.0,-6){\huge \color{black}...};

\Large
\node[draw,rectangle,rounded corners=3pt,blue] (a5)at(7.5,-7.5){\color{black}$\{1,2,3,4,5,6\}$};

\draw[-,>=latex,black!30,dashed] (e4) -- (e3);
\draw[-,>=latex,black!30,dashed] (e3) -- (e2);
\draw[-,>=latex,black!30,dashed] (e2) -- (e1);

\draw[-,>=latex,black!30,dashed] (e4) -- (a5);

\draw[-,>=latex,black!30,dashed] (a5) -- (a4.south);
\draw[-,>=latex,black!30,dashed] (a5) -- (b4);
\draw[-,>=latex,black!30,dashed] (a5) -- (d4);
\draw[-,>=latex,black!30,dashed] (a5) -- (f4);

\draw[-,>=latex,black!30,dashed] (b4) -- (d3);
\draw[-,>=latex,black!30,dashed] (d4) -- (b3);
\draw[-,>=latex,black!30,dashed] (d4) -- (e3);
\draw[-,>=latex,black!30,dashed] (e4) -- (f3);
\draw[-,>=latex,black!30,dashed] (f4) -- (e3);

\draw[-,>=latex,black!30,dashed] (a3) -- (a2);
\draw[-,>=latex,black!30,dashed] (b3) -- (b2);
\draw[-,>=latex,black!30,dashed] (b3) -- (c2);
\draw[-,>=latex,black!30,dashed] (c3) -- (c2);
\draw[-,>=latex,black!30,dashed] (e3) -- (d2);

\draw[-,>=latex,black!30,dashed] (a2) -- (d1);
\draw[-,>=latex,black!30,dashed] (b2) -- (c1);
\draw[-,>=latex,black!30,dashed] (c2) -- (c1);
\draw[-,>=latex,black!30,dashed] (d1) -- (c2);
\draw[-,>=latex,black!30,dashed] (d2) -- (d1);
\draw[-,>=latex,black!30,dashed] (e1) -- (d2);

\draw[-,>=latex,black!30,dashed] (c1) -- (a0);
\draw[-,>=latex,black!30,dashed] (d1) -- (a0);

\draw[-,>=latex,black!30,dashed] (e1) -- (a0);
\draw[-,>=latex,black!30,dashed] (a0) -- (b1);
\draw[-,>=latex,black!30,dashed] (b1) -- (c2);

\draw[->,>=latex,blue!90] (e4) -- (a5);
\draw[->,>=latex,blue!90] (a5) -- (c4);
\draw[->,>=latex,blue!90] (c4) -- (d3);
\draw[->,>=latex,blue!90] (d3) -- (c2);


            \end{tikzpicture}%
        }
\caption{An illustration of the MERGE-THEN-SPLIT strategy with a partial graph of 6 agents. The last best solution is the red node and the new best one is the green node. The nodes of the path between the two best solutions that visit the bottom node are represented by the blue nodes. } \label{mergethensplit}
\end{figure}

\begin{figure}[h!]
\centering
\normalsize
\resizebox{0.99\columnwidth}{8.9cm}{%
\begin{tikzpicture}
\large
\node[draw,rectangle,rounded corners=3pt,black] (a0)at(7.5,0){\color{black}$\{1\},\{2\},\{3\},\{4\},\{5\},\{6\}$};

\large
\node[] (k)at(-4.15,-1.5){\huge \color{black}...};
\node[draw,rectangle,rounded corners=3pt,black] (b1)at(-1.4,-1.5){\color{black}$\{1\},\{2\},\{4\},\{6\},\{3,5\}$};
\node[draw,rectangle,rounded corners=3pt,black] (c1)at(3.05,-1.5){\color{black}$\{3\},\{4\},\{5\},\{6\}\{1,2\}$};
\node[draw,rectangle,rounded corners=3pt,black] (d1)at(7.5,-1.5){\color{black}$\{1\},\{2\},\{4\},\{6\},\{3,5\}$};
\node[draw,rectangle,rounded corners=3pt,black] (e1)at(11.95,-1.5){\color{black}$\{2\},\{4\},\{5\},\{6\},\{1,3\}$};
\node[] (k)at(14.45,-1.5){\huge \color{black}...};

\large
\node[] (k)at(-5.0,-3){\huge \color{black}...};
\node[draw,rectangle,rounded corners=3pt,black] (a2)at(-2.5,-3){\color{black}$\{1\},\{2\},\{3,5\},\{4,6\}$};
\node[draw,rectangle,rounded corners=3pt,black] (b2)at(1.5,-3){\color{black}$\{3\},\{5\},\{6\},\{1,2,4\}$};
\node[draw,rectangle,rounded corners=3pt,vert] (c2)at(5.5,-3){\color{black}$\{4\},\{6\}\{1,2\},\{3,5\}$};
\node[draw,rectangle,rounded corners=3pt,black] (d2)at(9.5,-3){\color{black}$\{2\},\{4\},\{6\},\{1,3,5\}$};
\node[draw,rectangle,rounded corners=3pt,black] (e2)at(13.5,-3){\color{black}$\{4\},\{6\},\{1,3\},\{2,5\}$};
\node[] (k)at(16.0,-3){\huge \color{black}...};

\large
\node[] (k)at(-5.95,-4.5){\huge \color{black}...};
\node[draw,rectangle,rounded corners=3pt,black] (a3)at(-3.45,-4.5){\color{black}$\{1\},\{2\},\{3,4,5,6\}$};
\node[draw,rectangle,rounded corners=3pt,black] (b3)at(0.2,-4.5){\color{black}$\{6\},\{3,5\},\{1,2,4\}$};
\node[draw,rectangle,rounded corners=3pt,black] (c3)at(3.85,-4.5){\color{black}$\{1,2\},\{3,5\},\{4,6\}$};
\node[draw,rectangle,rounded corners=3pt,black] (d3)at(7.5,-4.5){\color{black}$\{4\},\{3,5\},\{1,2,6\}$};
\node[draw,rectangle,rounded corners=3pt,black] (e3)at(11.15,-4.5){\color{black}$\{6\},\{2,5\},\{1,3,4\}$};
\node[draw,rectangle,rounded corners=3pt,black] (f3)at(14.8,-4.5){\color{black}$\{2\},\{5\},\{1,3,4,6\}$};
\node[] (k)at(17.3,-4.5){\huge \color{black}...};

\Large
\node[] (k)at(-5.5,-6){\huge \color{black}...};
\node[draw,rectangle,rounded corners=3pt,black] (a4)at(-3.0,-6){\color{black}$\{1,2,3\},\{4,5,6\}$};
\node[draw,rectangle,rounded corners=3pt,black] (b4)at(0.5,-6){\color{black}$\{4\},\{1,2,3,5,6\}$};
\node[draw,rectangle,rounded corners=3pt,black] (c4)at(4.0,-6){\color{black}$\{1,2,6\},\{3,4,5\}$};
\node[draw,rectangle,rounded corners=3pt,black] (d4)at(7.5,-6){\color{black}$\{6\},\{1,2,3,4,5\}$};
\node[draw,rectangle,rounded corners=3pt,red] (e4)at(11.0,-6){\color{black}$\{2,5\},\{1,3,4,6\}$};
\node[draw,rectangle,rounded corners=3pt,black] (f4)at(14.5,-6){\color{black}$\{1,3,4\},\{2,5,6\}$};
\node[] (k)at(17.0,-6){\huge \color{black}...};

\Large
\node[draw,rectangle,rounded corners=3pt,black] (a5)at(7.5,-7.5){\color{black}$\{1,2,3,4,5,6\}$};

\draw[->,>=latex,blue!90] (e4) -- (e3);
\draw[->,>=latex,blue!90] (e3) -- (e2);
\draw[->,>=latex,blue!90] (e2) -- (e1);

\draw[-,>=latex,black!30,dashed] (e4) -- (a5);

\draw[-,>=latex,black!30,dashed] (a5) -- (a4.south);
\draw[-,>=latex,black!30,dashed] (a5) -- (b4);
\draw[-,>=latex,black!30,dashed] (a5) -- (d4);
\draw[-,>=latex,black!30,dashed] (a5) -- (f4);

\draw[-,>=latex,black!30,dashed] (b4) -- (d3);
\draw[-,>=latex,black!30,dashed] (d4) -- (b3);
\draw[-,>=latex,black!30,dashed] (d4) -- (e3);
\draw[-,>=latex,black!30,dashed] (e4) -- (f3);
\draw[-,>=latex,black!30,dashed] (f4) -- (e3);

\draw[-,>=latex,black!30,dashed] (a3) -- (a2);
\draw[-,>=latex,black!30,dashed] (b3) -- (b2);
\draw[-,>=latex,black!30,dashed] (b3) -- (c2);
\draw[-,>=latex,black!30,dashed] (c3) -- (c2);
\draw[-,>=latex,black!30,dashed] (e3) -- (d2);

\draw[-,>=latex,black!30,dashed] (a2) -- (d1);
\draw[-,>=latex,black!30,dashed] (b2) -- (c1);
\draw[-,>=latex,black!30,dashed] (c2) -- (c1);
\draw[->,>=latex,blue!90] (d1) -- (c2);
\draw[->,>=latex,blue!90] (d2) -- (d1);
\draw[->,>=latex,blue!90] (e1) -- (d2);

\draw[-,>=latex,black!30,dashed] (c1) -- (a0);
\draw[-,>=latex,black!30,dashed] (d1) -- (a0);

\draw[-,>=latex,black!30,dashed] (e1) -- (a0);
\draw[-,>=latex,black!30,dashed] (a0) -- (b1);
\draw[-,>=latex,black!30,dashed] (b1) -- (c2);

\draw[-,>=latex,black!30,dashed] (e4) -- (a5);
\draw[-,>=latex,black!30,dashed] (a5) -- (c4);
\draw[-,>=latex,black!30,dashed] (c4) -- (d3);
\draw[-,>=latex,black!30,dashed] (d3) -- (c2);


            \end{tikzpicture}%
        }
\caption{An illustration of the APPROACH-THEN-SWAP strategy with a partial graph of 6 agents. The last best solution is the red node and the new best one is the green node. The nodes of the path between the two best solutions are represented by the blue nodes. } \label{approachthenswap}
\end{figure}

\section{Hyperparameter Selection}

In order to optimize the performance of our algorithm, various hyperparameter values and strategies were tested and evaluated. We used four different methods for bridging paths and connecting solutions: SPLIT-THEN-MERGE, MERGE-THEN-SPLIT, APPROACH-THEN-SWAP, and a combination of all three. Additionally, we tested two strategies for resolving conflicts in multiagent search: bypassing conflicts and managing conflicts. For selecting child nodes, three strategies were considered: quantity-based selection, value-based selection, and random selection. Though generating child nodes randomly may seem simplistic, it has been shown to be an effective approach for many problems. Therefore, we implemented a random child node method that selects $\mathcal{N}_{a}$ child nodes and adds them to OPEN or RESERVE.

We also evaluated various values for the hyperparameters $\theta$ and $\omega$, including $\frac{n}{2}, n, 2n$, and $n^{2}$ for $\theta$, and 90\%, 99\%, 99.5\%, and 99.9\% for $\omega$. After testing multiple combinations of these strategies and hyperparameters, the best settings for the algorithm were determined to be: SPLIT-THEN-MERGE for bridging paths, managing conflicts for conflict resolution, value-based selection for selecting child nodes, $\theta = 2n$, and $\omega = 99.5\%$. These settings yielded the best results and optimal solutions more frequently.

As we can see in Table \ref{successratestrategies}, SALDAE with MERGE-THEN-SPLIT is able to find the optimal solution more often than with the other strategies. We conducted additional experiments varying the values of both $\theta$ and $\omega$ to find better results for these strategies. However, the behavior of the SPLIT-THEN-MERGE remained superior to the other strategies.

\begin{table*}[t]
\begin{center}
\renewcommand{\arraystretch}{1.5}
\small
\begin{tabular}{c|c|c|c}
\hline
\textbf{Distribution} &  \textbf{SPLIT-THEN-MERGE} & \textbf{MERGE-THEN-SPLIT} & \textbf{APPROACH-THEN-SWAP}\\
\hline
\textbf{Uniform} & 1580 (79\%) & 1402 (70.1\%) & 1509 (75.5\%) \\
\textbf{A-b Normal} & 582 (29.1\%) & 503 (25.2\%) & 561 (28.1\%) \\
\textbf{A-b Uniform} & 551 (27.5\%) & 532 (26.6\%) & 554 (27.7\%) \\
\textbf{Normal} & 519 (26\%) & 531 (26.6\%) & 521 (26.1\%) \\
\textbf{Beta} & 1485 (74.3\%) & 1392 (69.6\%) & 1462 (73.1\%) \\
\textbf{Exponential} & 1826 (91.3\%) & 1741 (87.1\%) & 1811 (90.6\%) \\
\textbf{Gamma} & 1628 (81.4\%) & 1623 (81.2\%) & 1673 (83.7\%) \\
\textbf{Pascal} & 2000 (100\%) & 2000 (100\%) & 2000 (100\%) \\
\textbf{Zipf} & 1286 (64.3\%) & 1263 (63.2\%) & 1287 (64.4\%) \\
\textbf{Cauchy} & 1714 (85.7\%) & 1722 (86.1\%) & 1710 (85.5\%) \\
\hline
\end{tabular}
\normalsize
\end{center}
\caption{Number of successes of the SALDAE algorithm on 2000 executions per distribution, when usinf the SPLIT-THEN-MERGE, MERGE-THEN-SPLIT, and APPROACH-THEN-SWAP strategies.}\label{successratestrategies}
\end{table*}


\section{Additional Results}

Table \ref{successrate} shows the number of times the algorithms obtain the optimal solution on 2000 executions per distribution. We ran the PICS and CSG-UCT algorithms to termination, and we stopped SALDAE at the time when ODP-IP found the optimal solution. The results show that SALDAE obtains the optimal solution in an average of $65.86\%$, whereas PICS and CSG-UCT find the optimal solution in an average of $17.36\%$ and $6.48\%$, respectively. For the Pascal distribution, SALDAE always obtains the optimal solution. For the Uniform, Normal, Agent-based Normal, and Agent-based Uniform, the solution quality is above 99\% when the optimal solution is not found. For the Beta distribution, the solution quality is above 99.9\% when the optimal solution is not found.

Figure \ref{solutionQuality1} shows additional results of SALDAE, PICS and CSG-UCT on other value distributions Beta, Agent-based Normal, Uniform, Agent-based Uniform, and Pascal.

\begin{figure}[h!]
\begin{center}
\small
\resizebox{0.32\columnwidth}{5.5cm}{%
     \begin{subfigure}{0.28\textwidth}
         \centering
         \begin{tikzpicture}
	\begin{axis}[	grid= major ,
            title=Beta ,
			width=0.8\textwidth ,
            xlabel = {Number of agents} ,
			ylabel = {Solution quality (\%)} ,
            width=5cm,height=5cm,
            xtick={4,6,8,10,12,14,16,18,20},
            xticklabels={4,6,8,10,12,14,16,18,20},
            ytick={70,80,90,95,99,100},
            yticklabels={70,80,90,95,99,100},
            ymin=90,
            ymax=102,
            label style={font=\large},
			title style={font=\Large},
			tick label style={font=\footnotesize},
            legend entries={SALDAE, CSG-UCT, PICS},
			legend style={at={(0,0.65)},anchor=north west,opacity=0.6,text opacity = 1}]
			\addplot+[only marks,color = blue,mark=square*,mark options={fill=cyan}] coordinates {(4,100) (5,100) (6,100) (7,100) (8,100) (9,100) (10,100) (11,100) (12,100) (13,100) (14,100) (15,100) (16,100) (17,99.97) (18,99.96) (19,99.99) (20,99.98) };//
			\addplot+[only marks,color = vert,mark=*,mark options={fill=vert}] coordinates {(4,100) (5,100) (6,100) (7,100) (8,100) (9,100) (10,100) (11,100) (12,100) (13,100) (14,99.9) (15,99.7) (16,99.2) (17,99.17) (18,98.76) (19,98.59) (20,98.68) };
            \addplot+[only marks,color = red,mark=triangle*,mark options={fill=red},ultra thin] coordinates {(4,100) (5,100) (6,100) (7,100) (8,100) (9,100) (10,100) (11,100) (12,100) (13,100) (14,100) (15,100) (16,99.9) (17,99.77) (18,99.76) (19,99.59) (20,99.68) };
	\end{axis}
    \end{tikzpicture}\\
         \label{fig:three sin x}
     \end{subfigure}
     }
\resizebox{0.32\columnwidth}{5.5cm}{%
     \begin{subfigure}{0.28\textwidth}
         \centering
         \begin{tikzpicture}
	\begin{axis}[	grid= major ,
            title=Agent-based Normal,
			width=0.6\textwidth ,
            xlabel = {Number of agents} ,
			ylabel = {Solution quality (\%)} ,
            width=5cm,height=5cm,
            xtick={4,6,8,10,12,14,16,18,20},
            xticklabels={4,6,8,10,12,14,16,18,20},
            ytick={70,80,90,95,99,100},
            yticklabels={70,80,90,95,99,100},
            ymin=94,
            ymax=101,
            label style={font=\large},
			title style={font=\Large},
			tick label style={font=\footnotesize},
            legend entries={SALDAE, CSG-UCT, PICS},
			legend style={at={(0,0.65)},anchor=north west,opacity=0.6,text opacity = 1}]
			\addplot+[only marks,color = blue,mark=square*,mark options={fill=cyan}] coordinates {(4,100) (5,100) (6,100) (7,100) (8,100) (9,100) (10,100) (11,100) (12,100) (13,100) (14,99.91) (15,99.97) (16,99.73) (17,99.47) (18,99.56) (19,99.44) (20,99.28) };//
			\addplot+[only marks,color = vert,mark=*,mark options={fill=vert}] coordinates {(4,100) (5,100) (6,100) (7,100) (8,100) (9,100) (10,100) (11,100) (12,100) (13,99.93) (14,99.94) (15,99.89) (16,99.66) (17,99.39) (18,99.46) (19,99.34) (20,99.18) };
            \addplot+[only marks,color = red,mark=triangle*,mark options={fill=red},ultra thin] coordinates {(4,100) (5,100) (6,100) (7,100) (8,100) (9,100) (10,100) (11,100) (12,100) (13,99.97) (14,99.92) (15,99.91) (16,99.63) (17,99.49) (18,99.46) (19,99.34) (20,99.38) };
	\end{axis}
    \end{tikzpicture}\\
         \label{fig:y equals x}
     \end{subfigure}
     }
     \resizebox{0.32\columnwidth}{5.5cm}{%
     \begin{subfigure}{0.28\textwidth}
         \centering
         \begin{tikzpicture}
	\begin{axis}[	grid= major ,
            title=Uniform ,
			width=0.6\textwidth ,
            xlabel = {Number of agents} ,
			ylabel = {Solution quality (\%)} ,
            width=5cm,height=5cm,
            xtick={4,6,8,10,12,14,16,18,20},
            xticklabels={4,6,8,10,12,14,16,18,20},
            ytick={70,80,90,95,99,100},
            yticklabels={70,80,90,95,99,100},
            ymin=90,
            ymax=102,
            label style={font=\large},
			title style={font=\Large},
			tick label style={font=\footnotesize},
            legend entries={SALDAE, CSG-UCT, PICS},
			legend style={at={(0,0.65)},anchor=north west,opacity=0.6,text opacity = 1}]
			\addplot+[only marks,color = blue,mark=square*,mark options={fill=cyan}] coordinates {(4,100) (5,100) (6,100) (7,100) (8,100) (9,100) (10,100) (11,100) (12,100) (13,100) (14,99.91) (15,99.98) (16,99.73) (17,99.47) (18,99.66) (19,99.64) (20,99.78) };//
			\addplot+[only marks,color = vert,mark=*,mark options={fill=vert}] coordinates {(4,100) (5,100) (6,100) (7,100) (8,100) (9,100) (10,99.8) (11,99.4) (12,99.1) (13,98.84) (14,98.61) (15,98.28) (16,98.43) (17,98.17) (18,98.06) (19,98.07) (20,98.03) };
            \addplot+[only marks,color = red,mark=triangle*,mark options={fill=red},ultra thin] coordinates {(4,100) (5,100) (6,100) (7,100) (8,100) (9,100) (10,100) (11,99.9) (12,99.8) (13,99.54) (14,99.11) (15,99.28) (16,99.13) (17,99.17) (18,99.06) (19,99.07) (20,99.03) };
	\end{axis}
    \end{tikzpicture}\\
         \label{fig:three sin x}
     \end{subfigure}
     }
     \resizebox{0.32\columnwidth}{5.5cm}{%
     \begin{subfigure}{0.25\textwidth}
         \centering
         \begin{tikzpicture}
	\begin{axis}[	grid= major ,
            title=(c) Pascal ,
			width=0.6\textwidth ,
            xlabel = {Number of agents} ,
			ylabel = {Solution quality (\%)} ,
            width=5cm,height=5cm,
            xtick={4,6,8,10,12,14,16,18,20},
            xticklabels={4,6,8,10,12,14,16,18,20},
            ytick={70,80,90,95,99,100},
            yticklabels={70,80,90,95,99,100},
            ymin=80,
            ymax=102,
            label style={font=\Large},
			title style={font=\Large},
			tick label style={font=\footnotesize},
            legend entries={SALDAE, CSG-UCT, PICS},
			legend style={at={(0,0.65)},anchor=north west,opacity=0.6,text opacity = 1}]
			\addplot+[only marks,color = blue,mark=square*,mark options={fill=cyan}] coordinates {(4,100) (5,100) (6,100) (7,100) (8,100) (9,100) (10,100) (11,100) (12,100) (13,100) (14,100) (15,100) (16,100) (17,100) (18,100) (19,100) (20,100) };//
			\addplot+[only marks,color = vert,mark=*,mark options={fill=vert}] coordinates {(4,100) (5,100) (6,100) (7,100) (8,96) (9,97) (10,94) (11,92) (12,88) (13,89) (14,86) (15,82.7) (16,83) (17,84) (18,82) (19,80.5) (20,81) };
            \addplot+[only marks,color = red,mark=triangle*,mark options={fill=red},ultra thin] coordinates {(4,100) (5,100) (6,100) (7,100) (8,100) (9,100) (10,100) (11,100) (12,100) (13,100) (14,100) (15,100) (16,100) (17,100) (18,100) (19,100) (20,100) };
	\end{axis}
    \end{tikzpicture}\\
         \label{fig:three sin x}
     \end{subfigure}
}
\:\:\:\:\:\:\:\:
\resizebox{0.32\columnwidth}{5.5cm}{%
     \begin{subfigure}{0.25\textwidth}
         \centering
         \begin{tikzpicture}
	\begin{axis}[	grid= major ,
            title=(d) Agent-based Uniform ,
			width=0.8\textwidth ,
            xlabel = {Number of agents} ,
			ylabel = {Solution quality (\%)} ,
            width=5cm,height=5cm,
            xtick={4,6,8,10,12,14,16,18,20},
            xticklabels={4,6,8,10,12,14,16,18,20},
            ytick={70,80,90,95,99,100},
            yticklabels={70,80,90,95,99,100},
            ymin=90,
            ymax=102,
            label style={font=\Large},
			title style={font=\Large},
			tick label style={font=\footnotesize},
            legend entries={SALDAE, CSG-UCT, PICS},
			legend style={at={(0,0.65)},anchor=north west,opacity=0.6,text opacity = 1}]
			\addplot+[only marks,color = blue,mark=square*,mark options={fill=cyan}] coordinates {(4,100) (5,100) (6,100) (7,100) (8,100) (9,100) (10,100) (11,100) (12,100) (13,100) (14,99.81) (15,99.27) (16,98.55) (17,97.87) (18,97.96) 
			(19,97.3) (20,97.13) };//
			\addplot+[only marks,color = vert,mark=*,mark options={fill=vert}] coordinates {(4,100) (5,100) (6,99.9) (7,99.8) (8,99.2) (9,98.7) (10,98.2) (11,96.9) (12,96.2) (13,95.6) (14,95.11) (15,94.27) (16,94.23) (17,93.77) (18,92.46) (19,92.54) (20,92.31) };
            \addplot+[only marks,color = red,mark=triangle*,mark options={fill=red},ultra thin] coordinates {(4,100) (5,100) (6,100) (7,100) (8,99.7) (9,99.4) (10,99.2) (11,98.9) (12,98.8) (13,98.2) (14,97.81) (15,97.27) (16,96.23) (17,96.77) (18,96.46) (19,96.54) (20,96.31) };
	\end{axis}
    \end{tikzpicture}\\
         \label{fig:three sin x}
     \end{subfigure}
     }

             \caption{Solution quality of SALDAE, PICS and CSG-UCT for sets of agents between 4 and 20.}
        \label{solutionQuality1}
\normalsize
\end{center}
\end{figure}
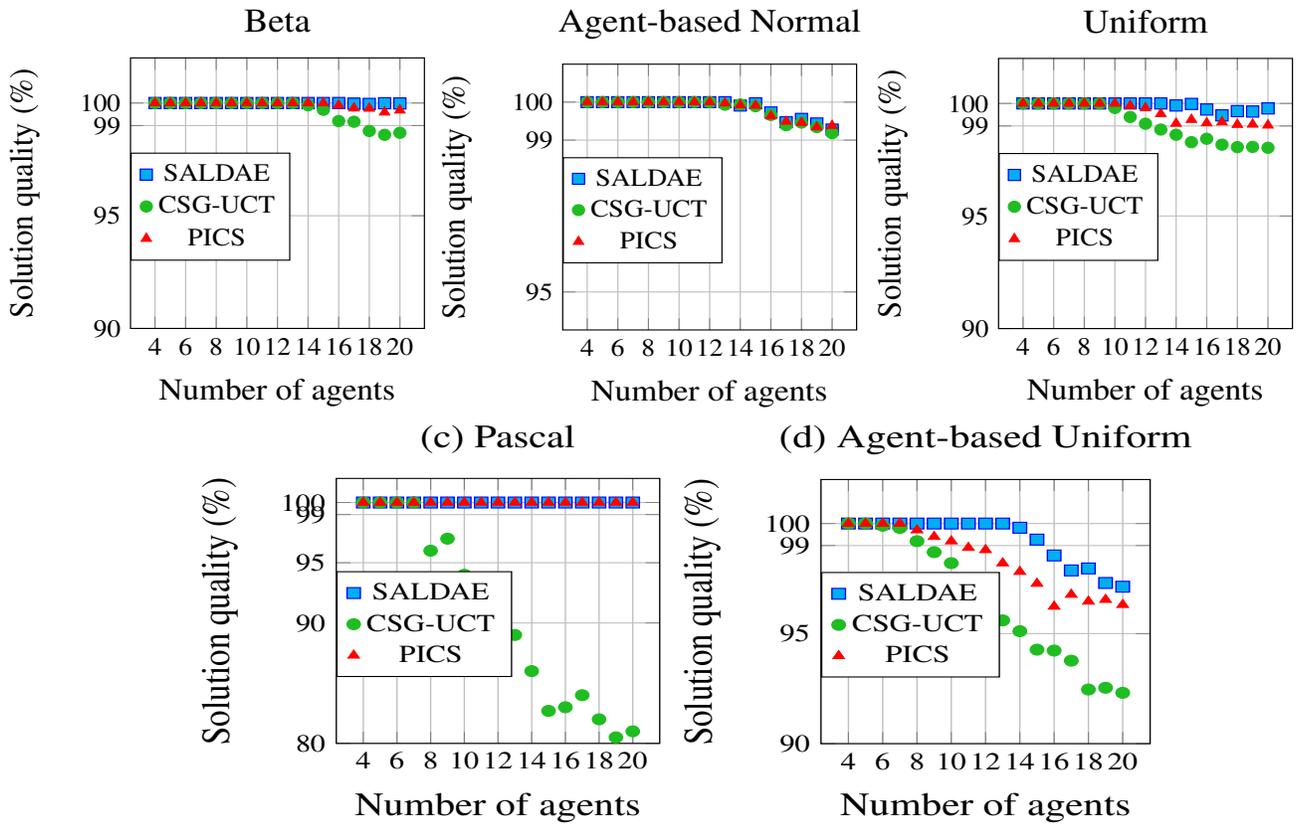

Figure \ref{largeScaleResults} shows additional results of SALDAE, PICS and CSG-UCT on other value distributions F, Uniform, Zipf, Exponential, and Normal.

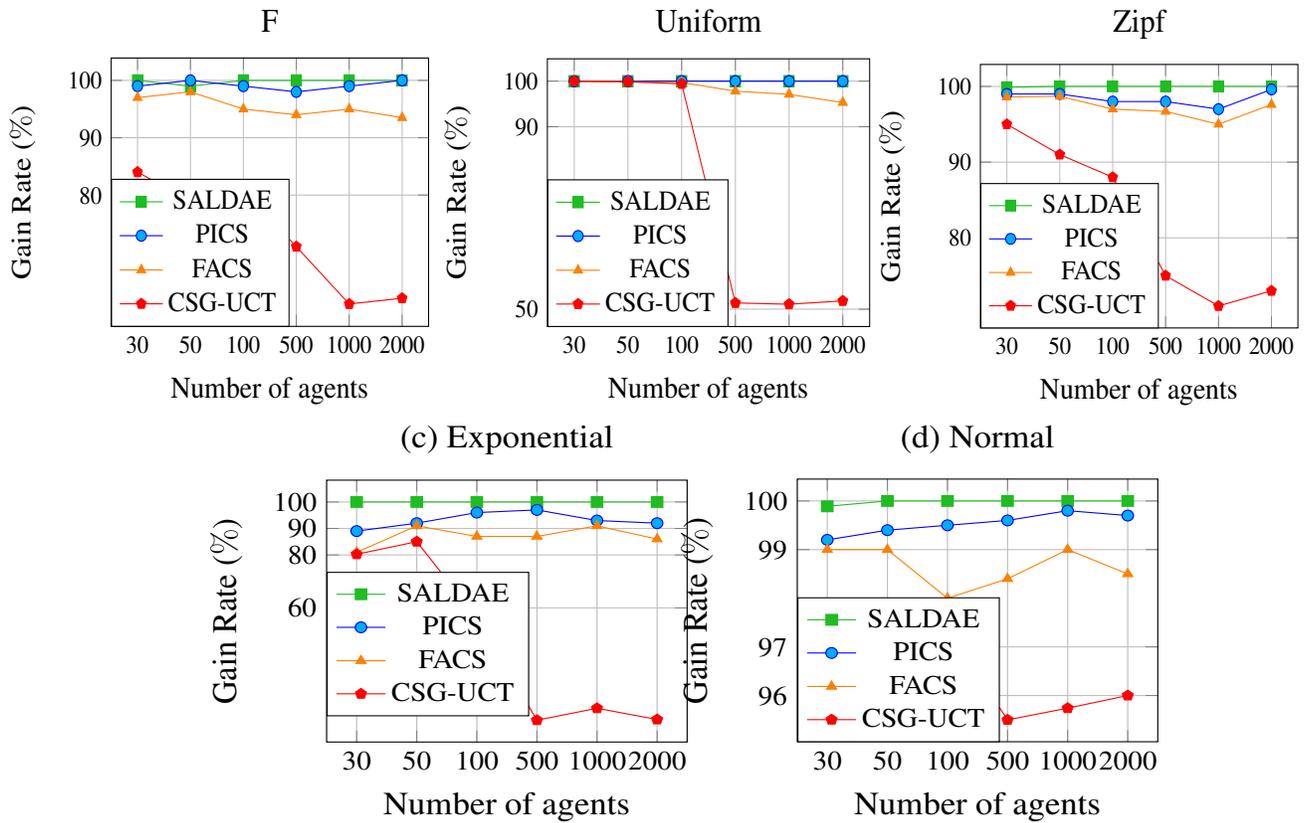
\begin{figure}[t]
\begin{center}
\normalsize
     \resizebox{0.32\columnwidth}{5.5cm}{%
     \begin{subfigure}{0.33\textwidth}
         \centering
         \begin{tikzpicture}
	\begin{axis}[	grid= major ,
            title= F ,
			width=0.8\textwidth ,
            xlabel = {Number of agents} ,
			ylabel = {Gain Rate ($\%$)} ,
            width=6cm,height=5cm,
			ytick={50,80,90,100},
            yticklabels={50,80,90,100},
            xtick={1,2,3,4,5,6},
            xticklabels={30,50,100,500,1000,2000},
            label style={font=\large},
			title style={font=\Large},
			tick label style={font=\small},
            legend entries={SALDAE, PICS, FACS, CSG-UCT},
            legend style={at={(0,0.55)},anchor=north west,opacity=0.6,text opacity = 1}]
            \addplot[color = vert,mark=square*,mark options={fill=vert}] coordinates {(1,100) (2,99) (3,100) (4,100) (5,100) (6,100) };
            \addplot[color = blue,mark=*,mark options={fill=cyan}] coordinates {(1,99) (2,100) (3,99) (4,98) (5,99) (6,100) };
            \addplot[color = orange,mark=triangle*,mark options={fill=orange}] coordinates {(1,97) (2,98) (3,95) (4,94) (5,95) (6,93.5) 
            };
            \addplot[mark=pentagon*,mark options={fill=red},color=red] coordinates {(1,84) (2,79) (3,81) (4,71) (5,61) (6,62)
            };
	\end{axis}
    \end{tikzpicture}\\
         \label{fig:three sin x}
     \end{subfigure}
}
\resizebox{0.32\columnwidth}{5.5cm}{%
     \begin{subfigure}{0.33\textwidth}
         \centering
         \begin{tikzpicture}
	\begin{axis}[	grid= major ,
            title= Uniform ,
			width=0.6\textwidth ,
            xlabel = {Number of agents} ,
			ylabel = {Gain Rate ($\%$)} ,
            width=6cm,height=5cm,
			ytick={50,90,100},
            yticklabels={50,90,100},
            xtick={1,2,3,4,5,6},
            xticklabels={30,50,100,500,1000,2000},
            label style={font=\large},
			title style={font=\Large},
			tick label style={font=\small},
            legend entries={SALDAE, PICS, FACS, CSG-UCT},
            legend style={at={(0,0.55)},anchor=north west,opacity=0.4,text opacity = 1}]
            \addplot[color = vert,mark=square*,mark options={fill=vert}] coordinates {(1,100) (2,100) (3,100) (4,100) (5,100) (6,100) };
            \addplot[color = blue,mark=*,mark options={fill=cyan}] coordinates {(1,99.97) (2,100) (3,100) (4,100) (5,100) (6,100) };
            \addplot[color = orange,mark=triangle*,mark options={fill=orange}] coordinates {(1,100) (2,99.9) (3,99.7) (4,97.8) (5,97.1) (6,95.3) 
            };
            \addplot[mark=pentagon*,mark options={fill=red},color=red] coordinates {(1,100) (2,99.86) (3,99.4) (4,51.36) (5,51.08) (6,51.81) };
	\end{axis}
    \end{tikzpicture}\\
         \label{fig:three sin x}
     \end{subfigure}
     }
\resizebox{0.32\columnwidth}{5.5cm}{%
     \begin{subfigure}{0.33\textwidth}
         \centering
         \begin{tikzpicture}
	\begin{axis}[	grid= major ,
            title= Zipf ,
			width=0.8\textwidth ,
            xlabel = {Number of agents} ,
			ylabel = {Gain Rate ($\%$)} ,
            width=6cm,height=5cm,
			ytick={50,80,90,100},
            yticklabels={50,80,90,100},
            xtick={1,2,3,4,5,6},
            xticklabels={30,50,100,500,1000,2000},
            label style={font=\large},
			title style={font=\Large},
			tick label style={font=\small},
            legend entries={SALDAE, PICS, FACS, CSG-UCT},
            legend style={at={(0,0.55)},anchor=north west,opacity=0.4,text opacity = 1}]
            \addplot[color = vert,mark=square*,mark options={fill=vert}] coordinates {(1,99.89) (2,100) (3,100) (4,100) (5,100) (6,100) };
            \addplot[color = blue,mark=*,mark options={fill=cyan}] coordinates {(1,99) (2,99) (3,98) (4,98) (5,97) (6,99.6) 
            };
            \addplot[color = orange,mark=triangle*,mark options={fill=orange}] coordinates {(1,98.6) (2,98.7) (3,97) (4,96.7) (5,95) (6,97.6) 
            };
            \addplot[mark=pentagon*,mark options={fill=red},color=red] coordinates {(1,95) (2,91) (3,88) (4,75) (5,71) (6,73)
            };
	\end{axis}
    \end{tikzpicture}\\
         \label{fig:three sin x}
     \end{subfigure}
}
\resizebox{0.32\columnwidth}{5.5cm}{%
     \begin{subfigure}{0.295\textwidth}
         \centering
         \begin{tikzpicture}
	\begin{axis}[	grid= major ,
            title=(c) Exponential ,
			width=0.8\textwidth ,
            xlabel = {Number of agents} ,
			ylabel = {Gain Rate ($\%$)} ,
            width=6cm,height=5cm,
			ytick={60,80,90,100},
            yticklabels={60,80,90,100},
            xtick={1,2,3,4,5,6},
            xticklabels={30,50,100,500,1000,2000},
            label style={font=\Large},
			title style={font=\Large},
            legend entries={SALDAE, PICS, FACS, CSG-UCT},
            legend style={at={(0,0.65)},anchor=north west,opacity=0.6,text opacity = 1}]
            \addplot[color = vert,mark=square*,mark options={fill=vert}] coordinates {(1,100) (2,100) (3,100) (4,100) (5,100) (6,100) };
            \addplot[color = blue,mark=*,mark options={fill=cyan}] coordinates {(1,89) (2,92) (3,96) (4,97) (5,93) (6,92) 
            };
            \addplot[color = orange,mark=triangle*,mark options={fill=orange}] coordinates {(1,81) (2,91) (3,87) (4,87) (5,91) (6,86) 
            };
            \addplot[mark=pentagon*,mark options={fill=red},color=red] coordinates {(1,80.26) (2,85.01) (3,56.29) (4,17.59) (5,22.09) (6,17.86) };
	\end{axis}
    \end{tikzpicture}\\
         \label{fig:y equals x}
     \end{subfigure}
     }
     \:\:\:\:\:
     \resizebox{0.32\columnwidth}{5.5cm}{%
     \begin{subfigure}{0.295\textwidth}
         \centering
         \begin{tikzpicture}
	\begin{axis}[	grid= major ,
            title=(d) Normal ,
			width=0.8\textwidth ,
            xlabel = {Number of agents} ,
			ylabel = {Gain Rate ($\%$)} ,
            width=6cm,height=5cm,
			ytick={90,96,97,99,100},
            yticklabels={90,96,97,99,100},
            xtick={1,2,3,4,5,6},
            xticklabels={30,50,100,500,1000,2000},
            label style={font=\Large},
			title style={font=\Large},
            legend entries={SALDAE, PICS, FACS, CSG-UCT},
            legend style={at={(0,0.55)},anchor=north west,opacity=0.4,text opacity = 1}]
            \addplot[color = vert,mark=square*,mark options={fill=vert}] coordinates {(1,99.89) (2,100) (3,100) (4,100) (5,100) (6,100) };
            \addplot[color = blue,mark=*,mark options={fill=cyan}] coordinates {(1,99.2) (2,99.4) (3,99.5) (4,99.6) (5,99.8) (6,99.7) 
            };
            \addplot[color = orange,mark=triangle*,mark options={fill=orange}] coordinates {(1,99) (2,99) (3,98) (4,98.4) (5,99) (6,98.5) 
            };
            \addplot[mark=pentagon*,mark options={fill=red},color=red] coordinates {(1,97.14) (2,96.53) (3,97.36) (4,95.50) (5,95.74) (6,96.0) };
	\end{axis}
    \end{tikzpicture}\\
         \label{fig:three sin x}
     \end{subfigure}
     }
    
             \caption{Gain rate of SALDAE, PICS, FACS and CSG-UCT when run with large numbers of agents.
             }
        \label{largeScaleResults}
\normalsize
\end{center}
\vspace{15pt}
\end{figure}

Figure \ref{solutionQuality10} shows the results of SALDAE using only one search agent. We compared the results to PICS using 20 search agents, and CSG-UCT, which is sequential. As can be seen, SALDAE performs better than PICS and CSG-UCT, while using only one search agent.

\begin{table}
\begin{center}
\renewcommand{\arraystretch}{1.5}
\small
\begin{tabular}{c|c|c|c}
\hline
\textbf{Distribution} &  \textbf{SALDAE} & \textbf{PICS} & \textbf{CSG-UCT}\\
\hline
\textbf{Uniform} & 1580 (79\%) & 153 (7.6\%) & 109 (5.4\%) \\
\textbf{A-b Normal} & 582 (29.1\%) & 132 (6.6\%) & 158 (7.9\%) \\
\textbf{A-b Uniform} & 551 (27.5\%) & 87 (4.4\%) & 97 (4.8\%) \\
\textbf{Normal} & 519 (26\%) & 98 (4.9\%) & 77 (3.8\%) \\
\textbf{Beta} & 1485 (74.3\%) & 502 (25.1\%) & 512 (25.6\%) \\
\textbf{Exponential} & 1826 (91.3\%) & 141 (7\%) & 101 (5\%) \\
\textbf{Gamma} & 1628 (81.4\%) & 192 (9.6\%) & 79 (4\%) \\
\textbf{Pascal} & 2000 (100\%) & 1743 (87\%) & 19 (1\%) \\
\textbf{Zipf} & 1286 (64.3\%) & 261 (13\%) & 97 (4.9\%) \\
\textbf{Cauchy} & 1714 (85.7\%) & 162 (8.1\%) & 47 (2.3\%) \\
\hline
\end{tabular}
\normalsize
\end{center}
\caption{Number of successes of the SALDAE, PICS and CSG-UCT algorithms on 2000 executions per distribution.}\label{successrate}
\end{table}

\begin{figure*}[h!]
\begin{center}
\small
\resizebox{0.31\textwidth}{5.7cm}{%
     \begin{subfigure}{0.28\textwidth}
         \centering
         \begin{tikzpicture}
	\begin{axis}[	grid= major ,
            title=Normal ,
			width=0.8\textwidth ,
            xlabel = {Number of agents} ,
			ylabel = {Solution quality (\%)} ,
            width=5cm,height=5cm,
            xtick={4,6,8,10,12,14,16,18,20},
            xticklabels={4,6,8,10,12,14,16,18,20},
            ytick={70,80,90,95,99,100},
            yticklabels={70,80,90,95,99,100},
            ymin=90,
            ymax=102,
            label style={font=\large},
			title style={font=\Large},
			tick label style={font=\footnotesize},
            legend entries={SALDAE (1), CSG-UCT, PICS},
			legend style={at={(0.8,1.6)},anchor=north}]
			\addplot+[only marks,color = blue,mark=square*,mark options={fill=cyan}] coordinates {(4,100) (5,100) (6,100) (7,100) (8,100) (9,100) (10,100) (11,100) (12,100) (13,100) (14,99.9) (15,99.7) (16,99.15) (17,99.17) (18,99.16) (19,99.24) (20,99.18) };//
			\addplot+[only marks,color = vert,mark=*,mark options={fill=vert}] coordinates {(4,100) (5,100) (6,100) (7,100) (8,99.7) (9,99.4) (10,99.2) (11,98.9) (12,98.8) (13,98.2) (14,97.81) (15,97.27) (16,96.23) (17,96.77) (18,96.46) (19,96.54) (20,96.31) };
            \addplot+[only marks,color = gold,mark=triangle*,mark options={fill=gold},ultra thin] coordinates {(4,100) (5,100) (6,100) (7,100) (8,100) (9,100) (10,100) (11,100) (12,100) (13,100) (14,99.9) (15,99.7) (16,99.2) (17,99.17) (18,98.76) (19,98.59) (20,98.68) };
            
	\end{axis}
    \end{tikzpicture}\\
         \label{fig:three sin x}
     \end{subfigure}
     }
\resizebox{0.31\textwidth}{5.7cm}{%
     \begin{subfigure}{0.28\textwidth}
         \centering
         \begin{tikzpicture}
	\begin{axis}[	grid= major ,
            title=Agent-based Normal,
			width=0.6\textwidth ,
            xlabel = {Number of agents} ,
			ylabel = {Solution quality (\%)} ,
            width=5cm,height=5cm,
            xtick={4,6,8,10,12,14,16,18,20},
            xticklabels={4,6,8,10,12,14,16,18,20},
            ytick={70,80,90,95,99,100},
            yticklabels={70,80,90,95,99,100},
            ymin=94,
            ymax=101,
            label style={font=\large},
			title style={font=\Large},
			tick label style={font=\footnotesize},
            legend entries={.,.,.},
			legend style={at={(1.75,1.5)},anchor=north}]
			\addplot+[only marks,color = blue,mark=square*,mark options={fill=cyan}] coordinates {(4,100) (5,100) (6,100) (7,100) (8,100) (9,100) (10,100) (11,100) (12,100) (13,99.95) (14,99.90) (15,99.87) (16,99.71) (17,99.42) (18,99.46) (19,99.4) (20,99.22) };//
			\addplot+[only marks,color = vert,mark=*,mark options={fill=vert}] coordinates {(4,100) (5,100) (6,100) (7,100) (8,100) (9,100) (10,100) (11,100) (12,100) (13,99.93) (14,99.94) (15,99.89) (16,99.66) (17,99.39) (18,99.46) (19,99.34) (20,99.18) };
            \addplot+[only marks,color = gold,mark=triangle*,mark options={fill=gold},ultra thin] coordinates {(4,100) (5,100) (6,100) (7,100) (8,100) (9,100) (10,100) (11,100) (12,100) (13,99.97) (14,99.92) (15,99.91) (16,99.63) (17,99.49) (18,99.46) (19,99.34) (20,99.38) };
	\end{axis}
    \end{tikzpicture}\\
         \label{fig:y equals x}
     \end{subfigure}
     }
     \resizebox{0.31\textwidth}{5.7cm}{%
     \begin{subfigure}{0.28\textwidth}
         \centering
         \begin{tikzpicture}
	\begin{axis}[	grid= major ,
            title=Uniform ,
			width=0.6\textwidth ,
            xlabel = {Number of agents} ,
			ylabel = {Solution quality (\%)} ,
            width=5cm,height=5cm,
            xtick={4,6,8,10,12,14,16,18,20},
            xticklabels={4,6,8,10,12,14,16,18,20},
            ytick={70,80,90,95,99,100},
            yticklabels={70,80,90,95,99,100},
            ymin=90,
            ymax=102,
            label style={font=\large},
			title style={font=\Large},
			tick label style={font=\footnotesize},
            legend entries={SALDAE (1), CSG-UCT, PICS},
			legend style={at={(0.2,1.6)},anchor=north}]
			\addplot+[only marks,color = blue,mark=square*,mark options={fill=cyan}] coordinates {(4,100) (5,100) (6,100) (7,100) (8,100) (9,100) (10,100) (11,100) (12,99.97) (13,99.9) (14,99.85) (15,99.78) (16,99.5) (17,99.27) (18,99.46) (19,99.34) (20,99.58) };//
			\addplot+[only marks,color = vert,mark=*,mark options={fill=vert}] coordinates {(4,100) (5,100) (6,100) (7,100) (8,100) (9,100) (10,99.8) (11,99.4) (12,99.1) (13,98.84) (14,98.61) (15,98.28) (16,98.43) (17,98.17) (18,98.06) (19,98.07) (20,98.03) };
            \addplot+[only marks,color = gold,mark=triangle*,mark options={fill=gold},ultra thin] coordinates {(4,100) (5,100) (6,100) (7,100) (8,100) (9,100) (10,100) (11,99.9) (12,99.8) (13,99.54) (14,99.11) (15,99.28) (16,99.13) (17,99.17) (18,99.06) (19,99.07) (20,99.03) };
	\end{axis}
    \end{tikzpicture}\\
         \label{fig:three sin x}
     \end{subfigure}
     }
     \hfill
     \resizebox{0.31\textwidth}{5.7cm}{%
     \begin{subfigure}{0.28\textwidth}
         \centering
         \begin{tikzpicture}
	\begin{axis}[	grid= major ,
            title=Beta ,
			width=0.8\textwidth ,
            xlabel = {Number of agents} ,
			ylabel = {Solution quality (\%)} ,
            width=5cm,height=5cm,
            xtick={4,6,8,10,12,14,16,18,20},
            xticklabels={4,6,8,10,12,14,16,18,20},
            ytick={70,80,90,95,99,100},
            yticklabels={70,80,90,95,99,100},
            ymin=90,
            ymax=102,
            label style={font=\large},
			title style={font=\Large},
			tick label style={font=\footnotesize},
            legend entries={SALDAE (1), CSG-UCT, PICS},
			legend style={at={(0.8,1.6)},anchor=north}]
			\addplot+[only marks,color = blue,mark=square*,mark options={fill=cyan}] coordinates {(4,100) (5,100) (6,100) (7,100) (8,100) (9,100) (10,100) (11,100) (12,100) (13,100) (14,100) (15,100) (16,100) (17,99.94) (18,99.93) (19,99.91) (20,99.95) };//
			\addplot+[only marks,color = vert,mark=*,mark options={fill=vert}] coordinates {(4,100) (5,100) (6,100) (7,100) (8,100) (9,100) (10,100) (11,100) (12,100) (13,100) (14,99.9) (15,99.7) (16,99.2) (17,99.17) (18,98.76) (19,98.59) (20,98.68) };
            \addplot+[only marks,color = gold,mark=triangle*,mark options={fill=gold},ultra thin] coordinates {(4,100) (5,100) (6,100) (7,100) (8,100) (9,100) (10,100) (11,100) (12,100) (13,100) (14,100) (15,100) (16,99.9) (17,99.77) (18,99.76) (19,99.59) (20,99.68) };
	\end{axis}
    \end{tikzpicture}\\
         \label{fig:three sin x}
     \end{subfigure}
     }
     \resizebox{0.31\textwidth}{5.7cm}{%
     \begin{subfigure}{0.28\textwidth}
         \centering
         \begin{tikzpicture}
	\begin{axis}[	grid= major ,
            title=Zipf ,
			width=0.6\textwidth ,
            xlabel = {Number of agents} ,
			ylabel = {Solution quality (\%)} ,
            width=5cm,height=5cm,
            xtick={4,6,8,10,12,14,16,18,20},
            xticklabels={4,6,8,10,12,14,16,18,20},
            ytick={70,80,90,95,99,100},
            yticklabels={70,80,90,95,99,100},
            ymin=90,
            ymax=102,
            label style={font=\large},
			title style={font=\Large},
			tick label style={font=\footnotesize},
            legend entries={.,.,.},
			legend style={at={(1.75,1.5)},anchor=north}]
			\addplot+[only marks,color = blue,mark=square*,mark options={fill=cyan}] coordinates {(4,100) (5,100) (6,100) (7,100) (8,100) (9,100) (10,100) (11,100) (12,100) (13,99.9) (14,99.71) (15,99.57) (16,99.21) (17,99.17) (18,99.06) (19,98.84) (20,98.71) };//
			\addplot+[only marks,color = vert,mark=*,mark options={fill=vert}] coordinates {(4,100) (5,100) (6,100) (7,100) (8,99.7) (9,99.4) (10,99.2) (11,98.5) (12,98.3) (13,97.8) (14,97.31) (15,96.57) (16,96.03) (17,95.77) (18,95.26) (19,95.14) (20,95.11) };
            \addplot+[only marks,color = gold,mark=triangle*,mark options={fill=gold},ultra thin] coordinates {(4,100) (5,100) (6,100) (7,100) (8,100) (9,100) (10,99.8) (11,99.7) (12,99.3) (13,99.1) (14,99.31) (15,98.07) (16,97.93) (17,97.77) (18,97.26) (19,97.14) (20,97.11) };
	\end{axis}
    \end{tikzpicture}\\
         \label{fig:three sin x}
     \end{subfigure}
     }
     \resizebox{0.31\textwidth}{5.7cm}{%
     \begin{subfigure}{0.28\textwidth}
         \centering
         \begin{tikzpicture}
	\begin{axis}[	grid= major ,
            title=Pascal ,
			width=0.6\textwidth ,
            xlabel = {Number of agents} ,
			ylabel = {Solution quality (\%)} ,
            width=5cm,height=5cm,
            xtick={4,6,8,10,12,14,16,18,20},
            xticklabels={4,6,8,10,12,14,16,18,20},
            ytick={70,80,90,95,99,100},
            yticklabels={70,80,90,95,99,100},
            ymin=80,
            ymax=102,
            label style={font=\large},
			title style={font=\Large},
			tick label style={font=\footnotesize},
            legend entries={SALDAE (1), CSG-UCT, PICS},
			legend style={at={(0.2,1.6)},anchor=north}]
			\addplot+[only marks,color = blue,mark=square*,mark options={fill=cyan}] coordinates {(4,100) (5,100) (6,100) (7,100) (8,100) (9,100) (10,100) (11,100) (12,100) (13,100) (14,100) (15,100) (16,100) (17,100) (18,100) (19,100) (20,100) };//
			\addplot+[only marks,color = vert,mark=*,mark options={fill=vert}] coordinates {(4,100) (5,100) (6,100) (7,100) (8,96) (9,97) (10,94) (11,92) (12,88) (13,89) (14,86) (15,82.7) (16,83) (17,84) (18,82) (19,80.5) (20,81) };
            \addplot+[only marks,color = gold,mark=triangle*,mark options={fill=gold},ultra thin] coordinates {(4,100) (5,100) (6,100) (7,100) (8,100) (9,100) (10,100) (11,100) (12,100) (13,100) (14,100) (15,100) (16,100) (17,100) (18,100) (19,100) (20,100) };
	\end{axis}
    \end{tikzpicture}\\
         \label{fig:three sin x}
     \end{subfigure}
}
\hfill
\resizebox{0.31\textwidth}{5.7cm}{%
     \begin{subfigure}{0.28\textwidth}
         \centering
         \begin{tikzpicture}
	\begin{axis}[	grid= major ,
            title=Agent-based Uniform ,
			width=0.8\textwidth ,
            xlabel = {Number of agents} ,
			ylabel = {Solution quality (\%)} ,
            width=5cm,height=5cm,
            xtick={4,6,8,10,12,14,16,18,20},
            xticklabels={4,6,8,10,12,14,16,18,20},
            ytick={70,80,90,95,99,100},
            yticklabels={70,80,90,95,99,100},
            ymin=90,
            ymax=102,
            label style={font=\large},
			title style={font=\Large},
			tick label style={font=\footnotesize},
            legend entries={SALDAE (1), CSG-UCT, PICS},
			legend style={at={(0.8,1.6)},anchor=north}]
			\addplot+[only marks,color = blue,mark=square*,mark options={fill=cyan}] coordinates {(4,100) (5,100) (6,100) (7,100) (8,100) (9,100) (10,100) (11,100) (12,100) (13,100) (14,99.81) (15,99.27) (16,98.55) (17,97.47) (18,97.56) 
			(19,96.3) (20,96.13) };//
			\addplot+[only marks,color = vert,mark=*,mark options={fill=vert}] coordinates {(4,100) (5,100) (6,99.9) (7,99.8) (8,99.2) (9,98.7) (10,98.2) (11,96.9) (12,96.2) (13,95.6) (14,95.11) (15,94.27) (16,94.23) (17,93.77) (18,92.46) (19,92.54) (20,92.31) };
            \addplot+[only marks,color = gold,mark=triangle*,mark options={fill=gold},ultra thin] coordinates {(4,100) (5,100) (6,100) (7,100) (8,99.7) (9,99.4) (10,99.2) (11,98.9) (12,98.8) (13,98.2) (14,97.81) (15,97.27) (16,96.23) (17,96.77) (18,96.46) (19,96.54) (20,96.31) };
	\end{axis}
    \end{tikzpicture}\\
         \label{fig:three sin x}
     \end{subfigure}
     }
     \resizebox{0.31\textwidth}{5.7cm}{%
     \begin{subfigure}{0.28\textwidth}
         \centering
         \begin{tikzpicture}
	\begin{axis}[	grid= major ,
            title=Gamma ,
			width=0.6\textwidth ,
            xlabel = {Number of agents} ,
			ylabel = {Solution quality (\%)} ,
            width=5cm,height=5cm,
            xtick={4,6,8,10,12,14,16,18,20},
            xticklabels={4,6,8,10,12,14,16,18,20},
            ytick={70,80,90,95,99,100},
            yticklabels={70,80,90,95,99,100},
            ymin=60,
            ymax=102,
            label style={font=\large},
			title style={font=\Large},
			tick label style={font=\footnotesize},
            legend entries={.,.,.},
			legend style={at={(1.75,1.5)},anchor=north}]
			\addplot+[only marks,color = blue,mark=square*,mark options={fill=cyan}] coordinates {(4,100) (5,100) (6,100) (7,100) (8,100) (9,100) (10,100) (11,100) (12,100) (13,99.9) (14,99.6) (15,99.59) (16,99.13) (17,99.27) (18,99.36) (19,99.34) (20,99.96) };//
			\addplot+[only marks,color = vert,mark=*,mark options={fill=vert}] coordinates {(4,100) (5,100) (6,100) (7,99.55) (8,97.1) (9,94.2) (10,91.96) (11,92.91) (12,89.86) (13,83.65) (14,84) (15,78.79) (16,73.63) (17,74.97) (18,71.16) (19,69.94) (20,65.1) };
            \addplot+[only marks,color = gold,mark=triangle*,mark options={fill=gold},ultra thin] coordinates {(4,100) (5,100) (6,100) (7,100) (8,100) (9,100) (10,99.96) (11,99.91) (12,99.86) (13,99.65) (14,99) (15,98.79) (16,97.63) (17,96.97) (18,96.16) (19,96.94) (20,96) };
	\end{axis}
    \end{tikzpicture}\\
         \label{fig:three sin x}
     \end{subfigure}
     }
     \resizebox{0.31\textwidth}{5.7cm}{%
     \begin{subfigure}{0.28\textwidth}
         \centering
         \begin{tikzpicture}
	\begin{axis}[	grid= major ,
            title=Exponential ,
			width=0.6\textwidth ,
            xlabel = {Number of agents} ,
			ylabel = {Solution quality (\%)} ,
            width=5cm,height=5cm,
            xtick={4,6,8,10,12,14,16,18,20},
            xticklabels={4,6,8,10,12,14,16,18,20},
            ytick={70,80,90,95,99,100},
            yticklabels={70,80,90,95,99,100},
            ymin=65,
            ymax=102,
            label style={font=\large},
			title style={font=\Large},
			tick label style={font=\footnotesize},
            legend entries={SALDAE (1), CSG-UCT, PICS},
			legend style={at={(0.2,1.6)},anchor=north}]
			\addplot+[only marks,color = blue,mark=square*,mark options={fill=cyan}] coordinates {(4,100) (5,100) (6,100) (7,100) (8,100) (9,100) (10,100) (11,99.9) (12,99.8) (13,99.2) (14,99.11) (15,98.68) (16,98.13) (17,97.87)
			(18,97.16) (19,96.64) (20,96.1) };//
			\addplot+[only marks,color = vert,mark=*,mark options={fill=vert}] coordinates {(4,100) (5,100) (6,100) (7,99.55) (8,97.1) (9,94.2) (10,91.96) (11,92.91) (12,89.86) (13,83.65) (14,84) (15,78.79) (16,73.63) (17,74.97) (18,71.16) (19,71.94) (20,70.8) };
            \addplot+[only marks,color = gold,mark=triangle*,mark options={fill=gold},ultra thin] coordinates {(4,100) (5,100) (6,100) (7,100) (8,99) (9,97) (10,94) (11,94) (12,93) (13,90) (14,89) (15,88.7) (16,87) (17,86.4) (18,86) (19,84.5) (20,84) };
	\end{axis}
    \end{tikzpicture}\\
         \label{fig:three sin x}
     \end{subfigure}
}

             \caption{Solution quality of SALDAE using one search agent, PICS and CSG-UCT for sets of agents between 4 and 20.}
        \label{solutionQuality10}
\normalsize
\end{center}
\end{figure*}
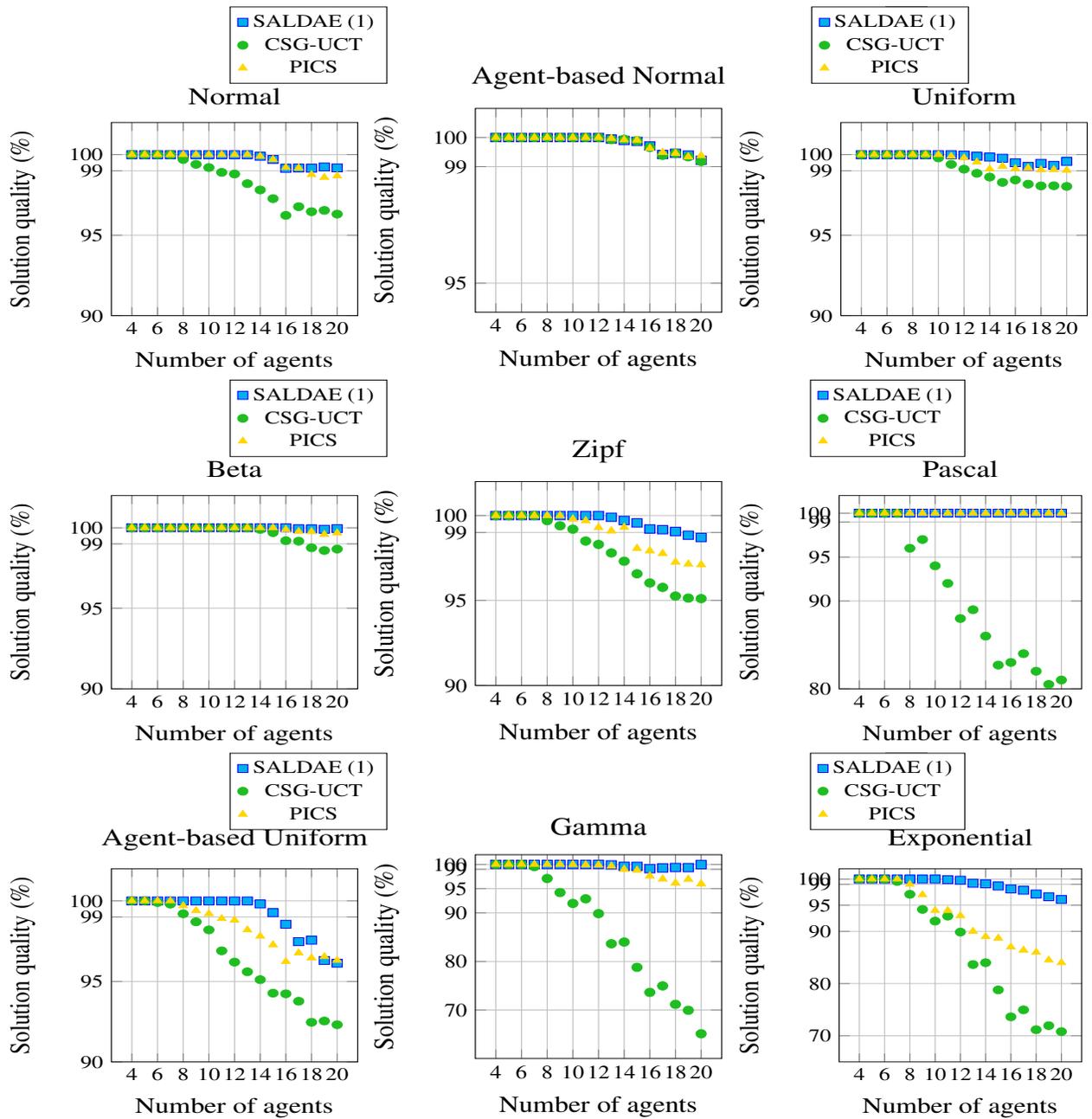

\end{document}